\documentclass[11pt]{article}
\usepackage[noparskip,biblatex]{tan}
\usepackage{morecommands}
\usepackage{changepage}

\title{Random dimension reduction and learning symmetric properties of quantum states}
\author{Angus Lowe\thanks{Authors are listed alphabetically.} \thanks{Center for Theoretical Physics --- a Leinweber Institute, MIT, Cambridge, MA, 02139. Email: \texttt{alowe7@mit.edu}.} \and Xinyu Tan\footnotemark[1] \thanks{Department of Mathematics, MIT, Cambridge, MA, 02139. Email: \texttt{norahtan@mit.edu}.}}
\date{}

\begin{document}

\maketitle
\begin{abstract}
We describe a procedure that simultaneously reduces the dimensions of many, potentially distinct quantum states while preserving properties invariant under the tensor power action of an isometry.
This serves as a black-box method to replace the dimension with the maximum rank in the sample complexity of learning symmetric properties, even those depending on multiple input states.
Combined with full state tomography, random dimension reduction yields improved upper bounds for estimating distances, fidelities, and relative entropies between pairs of states.

Expressing the action of this procedure through the Choi--Jamio\l{}kowski isomorphism reveals an intimate connection with the recently introduced random purification channel by Tang, Wright, and Zhandry. 
With this new perspective, we obtain a simple description of the rank-dependent version of the random purification channel and give an end-to-end analysis of sample-optimal tomography without requiring a reference to the Schur transform or Schur polynomials. 
Finally, we prove that there does not exist a random purification channel that simultaneously purifies copies of multiple, potentially different input states. 
Hence, the procedure is related to, but distinct from, random purification. 
\end{abstract}

% \tableofcontents

\section{Introduction}
Symmetries can often be leveraged to infer statistical properties of experimental data. One basic example of such a symmetry is \textit{label invariance}: the choice of which labels correspond to the various observations is immaterial in many cases.
Concretely, exchanging the role of one and zero for a random bit does not affect its binary entropy. Furthermore, this principle extends beyond the case where one is given independent samples from a \textit{single} distribution. Consider two different groups of patients undergoing the same treatment and exhibiting a range of possible outcomes: whether the outcomes of the groups are the same should not depend on the symbols used to record them, so long as we make the choice consistently across the two groups\footnote{Though, of course, conclusions about the efficacy of the treatment could be affected by a change of notation.}. Using this label invariance is the first step toward optimal algorithms for inferring \emph{symmetric} properties such as the Shannon entropy, support size, and distances between finite probability distributions~\cite{valiant2011estimating,Valiant2017estimatingentropy,han2018local,acharya2017unified}.

For quantum systems, a similar principle holds when the goal is to infer properties which are invariant under a change of basis. This includes spectral information such as von Neumann (and other) entropies~\cite{acharya2020estimating} or the rank~\cite{ow2015spectrum}, and --- though it appears not to have been fully exploited in the literature prior to our work --- distances between unknown states. Besides their fundamental importance as natural quantum analogues (and hence, generalizations) of the properties in classical learning theory mentioned above, such quantum properties have broad applications; for example, to the verification of quantum devices~\cite{elben2020cross}, probing entanglement in many-body systems~\cite{brydges2019probing,soleimanifar2022testing,lovitz2026nearly}, and analyzing quantum algorithms~\cite{chw2017weak}. Thus, understanding information-theoretic limits for these tasks is of both theoretical and technological interest.

Returning to the classical case, one salient feature of symmetric properties is the irrelevance of the ambient dimension of the distributions being sampled. Evidently, it is without loss of generality to discard all information about the labels or, equivalently, assign a fresh set of symbols $[k]\equiv \{1,2,\dots, k\}$ at the end of the sampling, so long as $k$ is chosen to be large enough to accommodate all distinct labels. Thus, labels lying outside the supports of the distributions have no bearing on the problem at hand, and this is true even when multiple unknown distributions are being sampled.

There is strong motivation to make an analogous statement for quantum states, where in many cases the ambient dimension of the unknown states far exceeds their ranks. However, the situation in the quantum case is more subtle and incurs at least two related challenges. The first is the absence of an operational interpretation in terms of discarding labels, upon which to guide one's intuition. The second challenge is in promoting the discrete group of label permutations to the continuous unitary group, which at its surface may seem to risk further complication. In this work, we show these challenges can be overcome, and that the above phenomenon extends to the quantum case quite generically.

\subsection{Our results}
\paragraph{Random dimension reduction.} We describe a quantum channel which simultaneously reduces the dimensions of many, potentially distinct quantum states while preserving symmetric properties; namely, properties invariant under the tensor power action of an isometry. 
More precisely, we show:
\begin{theorem}[Random dimension reduction]\label{thm:rand_joint_dim_red}
    Let $n,k,d$ be positive integers such that $1\leq k\leq d$, let $X_0$ be a linear operator on $(\C^k)^{\otimes n}$, and let $X$ be a linear operator on $(\C^d)^{\otimes n}$ such that $X=W^{\otimes n}X_0 (W^\dag)^{\otimes n}$ for some isometry $W\colon \C^k\to\C^d$. The quantum channel $\calR^{\mathrm{TP}}$, defined in \cref{eq:red_tp}, satisfies
    \begin{equation}
        \calR^{\mathrm{TP}}(X) = \E_{\bm{U}\sim \mathrm{U}(k)} \bU^{\otimes n}  X_0 (\bU^\dag)^{\otimes n}\label{eq:rand_joint_dim_red}
    \end{equation}
    where $\bU$ is a $k$-dimensional Haar-random unitary. Furthermore, $\calR^{\mathrm{TP}}$ can be implemented to $\epsilon$ accuracy in diamond norm using $\poly(n, \log d, \log(1/\epsilon))$ elementary gates. 
\end{theorem}
A useful case to think about is $k=d$, where the quantum channel $\calR^{\mathrm{TP}}$ is equal to the unitary twirling channel. When $k<d$, however, it behaves differently to remove the ambient dimension. 

We prove the first part of \Cref{thm:rand_joint_dim_red} in \Cref{sec:dim_red}; namely, that the twirling-like description of $\calR^{\mathrm{TP}}$ defined in \cref{eq:red_tp} satisfies \cref{eq:rand_joint_dim_red}. 
The proof uses only elementary linear algebra. 
Then in \Cref{sec:efficient} we make use of representation theory to show that the channel can be implemented efficiently on a quantum computer, through the Schur transform. 
This algorithm has previously been studied in the context of implementing universal decoders by Yoshida, Soeda, and Murao in~\cite[Eq.~(22)]{Yoshida2023universal}, where it is also shown that the resulting channel satisfies \cref{eq:rand_joint_dim_red}.
Thus, the second part of our proof establishes the equivalence between the elementary twirling-like description in \cref{eq:red_tp} and that in \cite{Yoshida2023universal}.

Our remaining goals in this work are to explore the application of the procedure to quantum learning theory, as well as its connection with the random purification channel~\cite{tang2025conjugatequerieshelp}.

\paragraph{Application to learning symmetric properties.}
Dimension reduction can be applied to simplify the setup in symmetric learning problems.
In the special case where one receives identical copies of a single unknown state (corresponding to $X=\rho^{\otimes n}$ in \Cref{thm:rand_joint_dim_red}), it is already well-known that measurements given by projections onto isotypic components of the symmetric group, known as \textit{Weak Schur sampling} (WSS), are optimal for unitarily invariant properties (see, e.g., Lemma 20 in \cite{montanaro2016survey}), and the resulting classical distribution has no dependence on the ambient dimension~\cite{HayashiMatsumoto02,Christandl2005,chw2017weak,odonnell2016efficient}. This leads to many state-of-the-art upper bounds for estimating symmetric properties.
Our procedure recovers this fact, and goes a step further: if $k$ is a known upper bound on the rank of the input state $\rho\in\C^{d\times d}$, then dimension reduction applied to $\rho^{\otimes n}$ results in a random state $\bm{\rho}_0^{\otimes n}$ where $\bm{\rho}_0\in\C^{k\times k}$ has the same nonzero eigenvalues as $\rho$.

\begin{table}[t]
    \centering
    %\footnotesize
    \renewcommand{\arraystretch}{1.25}
    \begin{tabular}{@{} >{\raggedright\arraybackslash}p{0.28\linewidth} >{\raggedright\arraybackslash}p{0.40\linewidth} >{\raggedright\arraybackslash}p{0.18\linewidth} @{}}
        \toprule
        \text{Estimation task} & \text{Best prior} & \makecell[l]{\textbf{This work}} \\
        \midrule
    
        Trace distance & $O(r^2\cdot\log^2(1/\epsilon)/\epsilon^4)$~\cite{WZ24a,chen2025list} & $O(r^2/\epsilon^2)$\\

        Fidelity & $O{(r^2\cdot\log^2(1/\epsilon)/\epsilon^4)}^*$~\cite{utsumi2025uhlmann} & $O(r^2/\epsilon^2)$ \\

        $\alpha$-Tsallis relative entropy, $0<\alpha<1/2$ & $O(r^{2+2\alpha}/\epsilon^{2/\alpha + 2/(1-\alpha)})$~\cite{chen2025list} & $O(r^2/\epsilon^{1/\alpha})$  \\

        Squared Hellinger distance & $O(r^3\cdot\log^2(r/\epsilon)/\epsilon^8)$~\cite{chen2025list} & $O(r^2/\epsilon^2)$ \\
        \bottomrule
    \end{tabular}

    \vspace{2pt}
    \caption{\label{tab:bounds}Comparison of the best prior sample complexity upper bounds with our results, which are obtained using dimension reduction followed by tomography. Here, $r$ denotes the maximum rank of the input states and $\epsilon$ is the additive error. The asterisk for the best prior result for fidelity is to indicate that the dependence is the same when $r$ is instead the minimum rank of the two states. The results for $\alpha$-Tsallis relative entropy are the same as for $(1-\alpha)$-Tsallis relative entropy when $0<\alpha<1/2$. If we set $\alpha=1/2$, this is exactly equal to the squared Hellinger distance.}
\end{table}

This is already a conceptual departure from previous statements in the literature. The main benefit of our approach, however, is that it extends beyond this special case: dimension reduction applies equally well to the setting where one has multiple, potentially different unknown quantum states, while preserving information about their relative eigenbases. 
This gives:
\begin{corollary}\label{cor:rho_sigma_cor}
    Let $\rho, \sigma \in \C^{d\times d}$ be quantum states such that $\rank(\rho) + \rank(\sigma) \leq k$. For any integers $a,b\geq 0$, the dimension reduction channel with $n=a+b$ in \Cref{thm:rand_joint_dim_red} satisfies
    \begin{equation*}
        \calR^{\mathrm{TP}}(\rho^{\otimes a} \otimes \sigma^{\otimes b}) = \E_{\bU\sim \mathrm{U}(k)}\left[ {(\bU \rho_0 \, \bU^\dagger)}^{\otimes a} \otimes {(\bU \sigma_0 \, \bU^\dagger)}^{\otimes b}\right]
    \end{equation*}
    where $\bm{U}$ is a $k$-dimensional Haar-random unitary and $\rho_0,\sigma_0\in \C^{k\times k}$ are quantum states such that $\rho=W\rho_0 W^\dag$ and $\sigma=W\sigma_0 W^\dag$ for some isometry $W\colon \C^k\to\C^d$. 
\end{corollary}
Note there is nothing particularly special about the case of two $d$-dimensional unknown states: one could apply dimension reduction to any number of such states or, by taking $b=0$, many copies of a single state. We choose to emphasize the case of two quantum states because this leads to improved upper bounds on the sample complexity of inferring several properties, including the trace distance, fidelity, Hellinger distance, and relative entropies, which we summarize along with the best prior work in \Cref{tab:bounds}. Here, the sample complexity refers to the total number of copies of either state.

The improvements are due to the following straightforward observation. In our context, we say that a property $f$ of two quantum states is symmetric if $f(\rho,\sigma) = f(W\rho W^\dag,W\sigma W^\dag)$ for any pair of states $\rho$ and $\sigma$ and isometry $W$. Suppose $f$ is also continuous, in that $F(\widehat{\rho},\rho)\geq 1-\delta$ and $F(\widehat{\sigma},\sigma)\geq 1-\delta$ implies $|f(\rho,\sigma) - f(\widehat{\rho},\widehat{\sigma})|\leq \epsilon$, where $F(\cdot,\cdot)$ is the standard fidelity function for quantum states.
Then, by \Cref{cor:rho_sigma_cor}, \emph{the sample complexity of estimating a symmetric function $f(\rho,\sigma)$ of two states $\rho$ and $\sigma$, each of rank at most $r$, to additive error $\epsilon$ is $O(r^2/\delta)$ via full state tomography on the dimension-reduced states}. The new upper bounds in \Cref{tab:bounds} follow from this statement along with straightforward continuity bounds applicable in each case, which we provide in \Cref{sec:continuity}.

These upper bounds, obtained with minimal reference to the structure of the problem at hand, should be viewed as a baseline against which to measure the performance of bespoke estimators. Remarkably, this approach seems to surpass an extensive list of prior work on such estimators in the low rank setting. For example, Refs.~\cite{WZ24a,WZ24b,utsumi2025uhlmann,bao2026estimatingquantumtsallisrelative,chen2025list} propose various estimators for such properties, improving upon prior work (e.g.,~\cite{gilyen2022improved}), by employing tools including density matrix exponentiation~\cite{Lloyd2014} and the QSVT~\cite{gilyen2019qsvt}. None of these fares better than the naive baseline given by dimension reduction and tomography, without additional assumptions beyond the maximum rank of the states (see \Cref{subsec:discussion}).

\paragraph{Connection to random purifications and tomography.}
The elementary analysis used to prove the existence of our procedure in \Cref{sec:dim_red} turns out to pay dividends through a connection to the recently introduced random purification channel of Tang, Wright, and Zhandry~\cite{tang2025conjugatequerieshelp}. For input states of the form $\rho^{\otimes n}$, the action of the dimension reduction channel is recovered by first applying this random purification channel, and then discarding the original system, leaving behind only the purifying registers. In the other direction, random purifications arise naturally by considering the action of the dimension reduction channel on a larger Hilbert space through the Choi--Jamio\l{}kowski isomorphism. This serves as a description of the rank-dependent random purification channel without referencing the Schur transform or representation theory, which was an open question in~\cite{girardi2026random}. It is also the last step in a line of work~\cite{pelecanos2025mixedstatetomographyreduces} giving an equally elementary analysis of sample-optimal quantum state tomography, which we provide for completeness in \Cref{sec:tomography}.

It is thus tantalizing to posit that dimension reduction and random purification are two sides of the same coin. However, in \Cref{sec:nogo} we show that this is not the case beyond input states of this form. In particular, there can be no single channel which, on the one hand, randomly purifies states of the form $\rho^{\otimes a}\otimes \sigma^{\otimes b}$ and which, on the other hand, reproduces the action of dimension reduction upon tracing out the original system. Hence, dimension reduction is a distinct procedure in all generality.

\subsection{Discussion}\label{subsec:discussion}
In the quantum learning literature, the Schur transform is almost exclusively applied to states of the form $\rho^{\otimes n}$. This is a well-understood procedure that serves as the first step in most learning algorithms that use fully entangled measurements. 
Our work shows that, perhaps surprisingly, it is also fruitful to apply the Schur transform to states of the form $\rho^{\otimes a}\otimes \sigma^{\otimes b}$ (see \Cref{alg:dim_red}). 
We are aware of only one prior instance of a similar observation: B\u{a}descu, O'Donnell, and Wright applied WSS to $\rho^{\otimes n} \otimes \sigma^{\otimes n}$ as part of an algorithm for estimating specific observables \cite[Theorem 7.10]{BOW19} in the context of quantum state certification.

Given that our techniques operate in the setting of fully entangled measurements, one may wonder whether our conceptual contribution --- that ambient dimension does not matter for the sample complexity of learning symmetric properties --- might also hold for less powerful measurements, such as single-copy measurements. However, this cannot be the case: even for pure states $\rho$ and $\sigma$, the task of estimating $\tr(\rho\sigma)$ to $\epsilon$ additive error with single-copy measurements requires $\Omega(\sqrt{d}/\epsilon)$ copies~\cite{ALL22}.

We conclude with several interesting directions for future research raised by this work:
\vspace{-.6em}
\paragraph{Minimum rank.}
Depending on the specific learning task, it is certainly possible to do better than our generic strategy of dimension reduction followed by full state tomography.
One prominent example is fidelity estimation, where we do not improve the previous best bounds in all regimes. 
In particular, prior work on this problem~\cite{gilyen2022improved,utsumi2025uhlmann} expresses the sample complexity in terms of the minimum rank, not the maximum. 
Such a dependence on the minimum rank is natural since $\sqrt{\rho}\sqrt{\sigma}$ has rank at most the minimum of $\rho$ and $\sigma$.
The same principle should hold for many other estimation tasks so long as they depend on quantities of the form $\rho^p \sigma^q$. 
We view improving bounds in this regime as an interesting open problem which will require new ideas beyond the generic method given in this work.

\vspace{-.6em}
\paragraph{Gate complexity.}
While we show that random dimension reduction can be implemented efficiently, we do not focus on optimizing gate counts.
So for certain estimation tasks, our gate complexity is higher than that of prior work, although we improve on the sample complexity.

\vspace{-.6em}
\paragraph{Robustness.} 
For both conceptual and practical reasons, it would be interesting to analyze the robustness of the procedure to cases where the input states are \emph{close to}, rather than exactly, low rank. Similar robustness results are known for state certification, e.g., in Corollary~1.6 of~\cite{BOW19}.

\subsection{Notation}
We let $\reg{A}, \reg{B},\dots$ etc denote quantum registers and $\calH_\reg{A}$, $\calH_{\reg{B}}$ be their corresponding Hilbert spaces. If there are $n$ copies of $\reg{A}$ then the corresponding Hilbert space is denoted $\calH_{\reg{A}^n}\cong\calH_{\reg{A}}^{\otimes n}$. We let $\mathrm{L}(\calH_{\reg{A}},\calH_{\reg{B}})$ denote the set of linear operators mapping $\calH_{\reg{A}}$ to $\calH_{\reg{B}}$, and $\mathrm{L}(\calH_{\reg{A}})$ the set of square linear operators on $\calH_{\reg{A}}$. We also let $\mathrm{D}(\calH_{\reg{A}})\subseteq \mathrm{L}(\calH_{\reg{A}})$ denote the set of density operators on $\calH_{\reg{A}}$. Suppose $\dim(\calH_{\reg{A}})=d$. We write $\mathrm{U}(\calH_{\reg{A}})$ for the unitary group acting on $\calH_{\reg{A}}$ and, if the space is clear from context, just $\mathrm{U}(d)$. We let $\mathfrak{S}_n$ denote the symmetric group on $n$ letters. Given $\pi \in \mathfrak{S}_n$, we write the permutation operator on $n$ copies of $\reg{A}$ as $F_{\pi}^{d}$ where its action in the standard basis of $\calH_{\reg{A}}^{\otimes n}$ is given by
  \begin{equation*}
      F_{\pi}^{d} \ket{i_1, \ldots, i_n} = \ket{i_{\pi^{-1}(1)}, \ldots, i_{\pi^{-1}(n)}} \quad \text{for every}\ i_1, \ldots, i_n \in [d] \,.
  \end{equation*}
If the dimension is clear from context we drop the superscript. Given a subalgebra $\calS\subseteq\mathrm{L}(\calH_{\reg{A}})$ we let $\calS^\prime$ denote its commutant, i.e., the set of all operators commuting with $\calS$.

We use bold font $\bm{x},\bm{y},\bm{V}$, etc for random variables, and $\bm{U}\sim \mathrm{U}(\calH_{\reg{A}})$ to mean $\bU$ is a Haar-random unitary acting on $\calH_{\reg{A}}$. We use $\mathrm{CP}$ and $\mathrm{TP}$ to refer to completely positive and trace-preserving, respectively. We let $X^+$ denote the Moore-Penrose pseudoinverse of the matrix $X$.

\section{Random dimension reduction}\label{sec:dim_red}
\subsection{Quantum information basics}

We record here some relevant preliminary facts from quantum information theory. 
\paragraph{Maximally entangled states.}
We let $\ket{\Gamma}_{\reg{AA}}\coloneq \sum_{i=1}^d\ket{i}_{\reg{A}}\otimes \ket{i}_{\reg{A}}$ be the unnormalized, maximally entangled state on two copies of $\reg{A}$.
\begin{lemma}[Maximally entangled state properties]\label{lem:max_ent_state_properties}
  Let $\calH_{\reg{A}}$ and $\calH_{\reg{B}}$ be finite-dimensional Hilbert spaces corresponding to quantum systems $\reg{A}$ and $\reg{B}$, respectively. The unnormalized maximally entangled state satisfies the following properties:
  \begin{enumerate}
    \item $(I_{\reg{A}}\otimes X)\ket{\Gamma}_{\reg{AA}} = (X^T\otimes I_\reg{B})\ket{\Gamma}_\reg{BB}$ for any $X\in \mathrm{L}(\calH_{\reg{A}},\calH_{\reg{B}})$; and
    \item $\tr_{\reg{A}}((X\otimes I_{\reg{A}})\Gamma_{\reg{AA}} (Y^\dag\otimes I_{\reg{A}})) = XY^\dag$ for any $X,Y\in \mathrm{L}(\calH_{\reg{A}},\calH_{\reg{B}})$.
  \end{enumerate}
\end{lemma}

\paragraph{Choi--Jamio\l{}kowski Isomorphism.} We adopt the convention where the Choi--Jamio\l{}kowski operator $J_\calN \in \mathrm{L}(\calH_{\reg{A}}\otimes \calH_{\reg{B}})$ of a quantum channel $\calN\colon \mathrm{L}(\calH_\reg{A})\to\mathrm{L}(\calH_\reg{B})$ is defined through
$
J_\calN \coloneq (\id_{\reg{A}}\otimes \calN)(\Gamma_{\reg{AA}}).
$
In particular, $\calN(X) = \tr_{\reg{A}}(J_\calN(X^T\otimes I))$.

\paragraph{Schur--Weyl duality.} We now briefly review some basic facts related to permutation and tensor power operations in finite dimensional linear algebra. See, e.g., Section~7.1 in~\cite{Watrous2018tqi}. Although useful for reasoning about irreducible representations as well, we defer these more fine-grained consequences to \Cref{sec:schur_and_efficient}. We let $\vee^n \calV$ denote the symmetric subspace of $\calV^{\otimes n}$; that is, the subspace of $\calV^{\otimes n}$ invariant under the natural action of the symmetric group $\mathfrak{S}_n$ on $\calV$. For example, the symmetric subspace of ${(\C^d)}^{\otimes n}$ is
\begin{align*}
  \vee^n\C^d\coloneq \{\ket{\psi}\in(\C^d)^{\otimes n}\colon F^d_\pi\ket{\psi}=\ket{\psi} , \ \forall\pi\in \mathfrak{S}_n\}
\end{align*}
while for the vector space $\mathrm{L}((\C^d)^{\otimes n})$ we have
\begin{align*}
  \vee^n\mathrm{L}(\C^d)\coloneq \{X\in\mathrm{L}((\C^d)^{\otimes n})\colon F^d_\pi X(F^d_\pi)^\dag = X , \ \forall \pi\in \mathfrak{S}_n\} = \{F^d_\pi\colon \pi\in \mathfrak{S}_n\}^\prime.
\end{align*}
We refer to $\vee^n\mathrm{L}(\C^d)$ as the set of \emph{exchangeable operators} on $\mathrm{L}((\C^d)^{\otimes n})$. Clearly, the algebra generated by the tensor power action of the unitary group $\{U^{\otimes n}\colon U\in \mathrm{U}(d)\}$ is contained in $\vee^n\mathrm{L}(\C^d)$.\; \emph{Schur--Weyl duality} refers to the fact that these algebras are actually equal. Conversely, by von Neumann's bicommutant theorem, the commutant $\{U^{\otimes n}\colon U\in \mathrm{U}(d)\}^\prime$ is the algebra generated by permutation operators $\{F_\pi^d\colon \pi\in \mathfrak{S}_n\}$. 

\subsection{Defining the quantum channel}
Let $\calH_\reg{A}= \C^d$ and $\calH_{\reg{B}}= \C^k$ for some $1\leq k\leq d$, and let $W_0\coloneq \sum_{i\in [k]}{|i\rangle}_{\reg{A}}{\langle i |}_{\reg{B}}$ be the natural isometric embedding of $\calH_{\reg{B}}$ in $\calH_\reg{A}$. By a \emph{Haar random isometry} we mean the isometry-valued random variable
$
    \bm{W} \coloneq \bm{V} W_0
$
where $\bm{V}\in \mathrm{V}(\calH_\reg{A})$ is Haar-random. Here and throughout, let $n$ be a positive integer which denotes the number of identical copies of $\reg{A}$ comprising our system of interest. We then let
\begin{align}\label{eq:m_defn}
    M^{(n)}_{\reg{B}\embed\reg{A}} = \E_{\bm{W}} \paren[\big]{\bm{W}\bm{W}^\dag}^{\otimes n}
\end{align}
denote the expected value of the projector onto the image of $\bm{W}^{\otimes n}$. This is no more than the $n$th moment of a random $k$-dimensional projection taken uniformly over the Grassmannian of $k$-dimensional subspaces in $\C^d$. We also use dimensions in the subscript to denote the same operator, e.g., $M^{(n)}_{k\embed d}$ in this case, and when $n$, $\reg{A}$, and $\reg{B}$ are clear from context, just $M$. 

We define the \emph{dimension reduction map} $\calR: \mathrm{L}(\calH_{\reg{A}^n})\to \mathrm{L}(\calH_{\reg{B}^n})$ through
\begin{align*}
  \calR(X) = \E \left[ (\bm{W}^\dag)^{\otimes n}\sqrt{M^+} X \sqrt{M^+} \bm{W}^{\otimes n}\right]
\end{align*}
for every $X\in\mathrm{L}(\calH_{\reg{A}^n})$, where $\bm{W}\colon \calH_\reg{B}\to\calH_\reg{A}$ is a Haar-random isometry. Evidently, this map is CP and, although it is not TP on the full domain, it \emph{is} CPTP when restricted to operators supported entirely on the image of $M$. Hence, $\calR$ can straightforwardly be modified to a valid quantum channel $\calR^\mathrm{TP}$ through\footnote{Although we use the maximally mixed state on $\calH_{\reg{B}^n}$ in the definition, any fixed quantum state will do.}
\begin{align}\label{eq:red_tp}
  \calR^{\mathrm{TP}}(X) = \calR(X) + \tr(QX)\frac{I_{\reg{B}^n}}{k^n}
\end{align}
for every $X\in\mathrm{L}(\calH_{\reg{A}^n})$, where $Q$ is the projector onto the kernel of $M$. This map has the desirable property that the dimensions of each copy of the system $\reg{A}$ are reduced while preserving all properties symmetric under the tensor power action of an isometry from dimension $k$ to $d$, as stated in \cref{eq:rand_joint_dim_red} within \Cref{thm:rand_joint_dim_red}. We prove this in the following section.

\subsection{Proof of dimension reduction}\label{sec:proof_of_dim_red}
In this section, we prove the first part of \Cref{thm:rand_joint_dim_red}; namely, that the dimension reduction channel $\calR^{\mathrm{TP}}$ (defined in \cref{eq:red_tp}) has the action given in \cref{eq:rand_joint_dim_red}. We defer the proof of an efficient implementation to \Cref{sec:efficient}.

We restate this claim now for the convenience of the reader: If $X\in \mathrm{L}(\calH_{\reg{A}^n})$ and $X_0\in\mathrm{L}(\calH_{\reg{B}^n})$ satisfy $X=W^{\otimes n}X_0 (W^\dag)^{\otimes n}$ for some isometry $W\colon \calH_{\reg{B}}\to \calH_{\reg{A}}$ then
\begin{equation*}
        \calR^{\mathrm{TP}}(X) = \E_{\bm{U}\sim \mathrm{U}(\calH_{\reg{B}})} \bU^{\otimes n}  X_0 (\bU^\dag)^{\otimes n}.\tag{\cref{eq:rand_joint_dim_red}}
\end{equation*}
We first note the following straightforward lemmas.
\begin{lemma}\label{lem:supp_of_c_op}
  If $X= W^{\otimes n}X_0 (W^\dag)^{\otimes n}$ for some $X_0\in\mathrm{L}(\calH_{\reg{B}^n})$ and fixed isometry $W\colon\calH_{\reg{B}}\to\calH_{\reg{A}}$ then $QX = XQ = 0$, where $Q$ is the projector onto the kernel of $M$.
\end{lemma}
\begin{proof}
 If $\ket{\psi}\in \ker(M)$ then $\bra{\psi}M\ket{\psi}=\bra{\psi}\E(\bm{W} \bm{W}^\dag)^{\otimes n}\ket{\psi}=\E\norm{(\bm{W}^\dag)^{\otimes n}\ket{\psi}}^2 = 0$. Suppose by way of contradiction that $\norm{(W_0^\dag U^\dag)^{\otimes n}\ket{\psi}}^2>0$ for some $U\in \mathsf{U}(\calH_{\reg{A}})$. Then, since the Haar measure has full support and the function $U\mapsto \norm{(W_0^\dag U^\dag)^{\otimes n} \ket{\psi}}^2$ is continuous, $\E \norm{(W_0^\dag \bm{U}^\dag)^{\otimes n}\ket{\psi}}^2 = \E \norm{(\bm{W}^\dag)^{\otimes n}\ket{\psi}}^2>0$. Therefore, we must have $(W^\dag)^{\otimes n} \ket{\psi}=0$. Hence, $\ker(M)\subseteq \ker({(W^\dag)^{\otimes n}})$, and therefore $QW^{\otimes n}=((W^\dag)^{\otimes n} Q)^\dag=0$. 
\end{proof}
\begin{lemma}\label{lem:c-commutes}
    For any $\pi\in \mathfrak{S}_n$ it holds that $[M^{(n)}_{\reg{B}\embed \reg{A}}, F^{d}_\pi]=0$. 
\end{lemma}
\begin{proof}
     For any linear operator $Y\colon \calH_\reg{B}\to\calH_\reg{A}$, the fact that $Y^{\otimes n}F^{k}_\pi = F^{d}_\pi Y^{\otimes n}$ and $(Y^\dag)^{\otimes n}F^d_\pi = F^k_\pi (Y^\dag)^{\otimes n}$ is immediate. Therefore, the permutation operator $F_\pi^d$ commutes with any value of the random operator $(\bU W_0W_0^\dag \bU^\dag)^{\otimes n}$ appearing in the average that defines $M^{(n)}_{\reg{B}\embed \reg{A}}$ in \cref{eq:m_defn}.
\end{proof}

We are now ready to prove \cref{eq:rand_joint_dim_red}. 
\begin{proof}[Proof of \cref{eq:rand_joint_dim_red}]
By \Cref{lem:supp_of_c_op} we have $\calR(X)=\calR^{\mathrm{TP}}(X)$, so it suffices to prove the action of $\calR$ on $X$ is as claimed. Since a Haar-random isometry $\bm{W}$ is identically distributed as $\bV W_0$ for Haar-random $\bV\in\mathrm{U}(\calH_{\reg{A}})$, we have that
\begin{align*}
\calR(X) &= \E_{\bm{V}\sim\mathrm{U}(\calH_\reg{A})} \left[(W_0^\dag\bm{V}^\dag)^{\otimes n} \sqrt{M^+} X \sqrt{M^+}(\bm{V}W_0)^{\otimes n} \right] \\
&= (W_0^\dag)^{\otimes n} \underbrace{\E_{\bm{V}\sim\mathrm{U}(\calH_\reg{A})} \left[(\bm{V}^\dag)^{\otimes n} \sqrt{M^+} X \sqrt{M^+}\bm{V}^{\otimes n} \right]}_{\text{the Haar average}}W_0^{\otimes n} \,.
\end{align*}
Notice that the Haar average is in $\Span\{F_{\pi}^d: \pi \in \mathfrak{S}_n\}$. Since $(W_0^\dagger)^{\otimes n} F_{\pi}^d W_0^{\otimes n} = F_{\pi}^k$ for any $\pi \in \mathfrak{S}_n$, we have that $\calR(X) \in \Span\{F_{\pi}^k: \pi \in \mathfrak{S}_n\}$. 

Set $Y = \E_{\bm{U}\sim \mathrm{U}(\calH_{\reg{B}})} \bU^{\otimes n}  X_0 (\bU^\dag)^{\otimes n}$, which is also in $\Span\{F_{\pi}^k: \pi \in \mathfrak{S}_n\}$ due to the Haar-twirling. 
To prove $\calR(X) = Y$, it suffices to show that $\tr[(F_{\pi}^k)^\dagger \calR(X)] = \tr[(F_{\pi}^k)^\dagger Y]$ for every $\pi \in \mathfrak{S}_n$. Indeed,
\begin{align}
    \tr[(F_{\pi}^k)^\dagger \calR(X)] &= \E_{\bW} \tr \left[ (F_{\pi}^k)^\dagger (\bm{W}^\dag)^{\otimes n} \sqrt{M^+} X \sqrt{M^+}\bm{W}^{\otimes n} \right] \nonumber\\
    &= \E_{\bW} \tr \left[ \sqrt{M^+}\bm{W}^{\otimes n}(F_{\pi}^k)^\dagger (\bm{W}^\dag)^{\otimes n} \sqrt{M^+} X  \right] \tag{cyclicity of trace}\\
    &= \E_{\bW} \tr \left[ \sqrt{M^+}(F_{\pi}^d)^\dagger (\bW\bm{W}^\dag)^{\otimes n} \sqrt{M^+} X  \right] \tag{$\bW^{\otimes n} (F_{\pi}^k)^\dagger = (F_{\pi}^d)^\dagger \bW^{\otimes n}$} \\
    &= \tr \left[ \sqrt{M^+}(F_{\pi}^d)^\dagger M \sqrt{M^+} X  \right] \nonumber\\
    &= \tr \left[ (F_{\pi}^d)^\dagger \sqrt{M^+} M \sqrt{M^+} X  \right] \tag{\Cref{lem:c-commutes}} \\
    &= \tr [ (F_{\pi}^d)^\dagger X  ]  \nonumber \\
    &= \tr [ (F_{\pi}^k)^\dagger X_0  ] \,, \nonumber
\end{align}
where in the second-to-last equality we used $\sqrt{M^+}M\sqrt{M^+} = P$ as well as the fact established in \Cref{lem:supp_of_c_op} that the support of $X$ lies in the image of $M$, so $PX = X$. On the other hand,
\begin{equation*}
    \tr[(F_{\pi}^k)^\dagger Y] = \E_{\bm{U}\sim \mathrm{U}(\calH_{\reg{B}})} \tr[(F_{\pi}^k)^\dagger \bU^{\otimes n}  X_0 (\bU^\dag)^{\otimes n}] = \tr [ (F_{\pi}^k)^\dagger X_0  ] \,. 
\end{equation*}
This completes the proof of \cref{eq:rand_joint_dim_red}.
\end{proof}

\section{Relation to random purification and tomography}\label{sec:acorn}

\subsection{Random purification from dimension reduction}\label{sec:connection_description}
In this section we explain the connection between random purifications and the dimension reduction channel, culminating in \Cref{thm:rand_pur}. We will not require the representation-theoretic tools from \Cref{sec:schur_and_efficient}.

It is insightful to consider the Choi--Jamio\l{}kowski operator for the dimension reduction map, $J_\calR\in \mathrm{L}(\calH_{\reg{A}^n}\otimes \calH_{\reg{B}^n})$, which is equal to
\begin{equation*}
  J_\calR = \E\left[ (I_{\reg{A}^n}\otimes (\bm{W}^\dag)^{\otimes n}\sqrt{M^+})\Gamma_{\reg{AA}}^{\otimes n}(I_{\reg{A}^n}\otimes \sqrt{M^+}\bm{W}^{\otimes n})\right] \,.
\end{equation*}
The other direction for the Choi--Jamio\l{}kowski isomorphism is given by the well-known equality $\calR(X)=\tr_{\reg{A}^n}(J_\calR\cdot (X^T\otimes I_{\reg{B}^n}))$ for any $X\in\mathrm{L}(\calH_{\reg{A}^n})$ (see, e.g., \cite[Section~2.2.2]{Watrous2018tqi}).
This representation for the action of $\calR$, derived from a larger Hilbert space before discarding the $\reg{A}$ registers, is significant in its own right.
Indeed, let us define the \emph{random purification extension} $\calP: \mathrm{L}(\calH_{\reg{A}^n})\to \mathrm{L}(\calH_{\reg{A}^n}\otimes \calH_{\reg{B}^n})$ through
\begin{align}\label{eq:random_pur_ext}
  \calP(X) = J_{\calR}\cdot (X\otimes I_{\reg{B}^n})
\end{align}
for every $X\in\mathrm{L}(\calH_{\reg{A}^n})$.

The use of the word ``extension'' reflects the fact that, while $\calP$ is not CP on the full domain, it \emph{is} a noteworthy CP map when restricted to exchangeable operators, as we show in \Cref{thm:rand_pur} below. This implies it can be modified to a CP map on the full domain by first twirling over the symmetric group: Let $\calT_{\mathfrak{S}_n}\colon \mathrm{L}(\calH_{\reg{A}^n})\to \vee^n \mathrm{L}(\calH_{\reg{A}})$ denote the permutation twirl $\calT_{\mathfrak{S}_n}\colon X\mapsto \frac{1}{n!}\sum_{\pi\in \mathfrak{S}_n}F_\pi X F_\pi^\dag$, which is a quantum channel that fixes every exchangeable operator. We define $\calP^{\mathrm{CP}}\colon \mathrm{L}(\calH_{\reg{A}^n})\to\mathrm{L}(\calH_{\reg{A}^n}\otimes \calH_{\reg{B}^n})$ through
\begin{align*}
  \calP^{\mathrm{CP}}(X) =\calP(\calT_{\mathfrak{S}_n}(X))
\end{align*}
for every $X\in\mathrm{L}(\calH_{\reg{A}^n})$. As with the dimension reduction map, $\calP^{\mathrm{CP}}$ can be modified to be TP as well, for instance, through
\begin{align}\label{eq:tp_rand_pur}
  \calP^{\mathrm{CPTP}}(X) = \calP^{\mathrm{CP}}(X) + \tr(Q X)\frac{I_{\reg{A}^n}\otimes I_{\reg{B}^n}}{k^nd^n}
\end{align} 
for every $X\in\mathrm{L}(\calH_{\reg{A}^n})$, with $Q$ the projector onto the kernel of $M$ (as in \cref{eq:red_tp}). 
Moreover, the actions of all three maps $\calP$, $\calP^{\mathrm{CP}}$, and $\calP^{\mathrm{CPTP}}$ coincide on exchangeable operators supported on the image of $M^{(n)}_{k\embed d}$; for example, when acting on $n$ copies of a quantum state of rank at most $k$.
In this case, they have the same action as the random purification channel (AKA acorn trick) originally introduced in \cite{tang2025conjugatequerieshelp}.

\begin{theorem}[Random purification from dimension reduction]\label{thm:rand_pur}
  It holds that $\calP^{\mathrm{CPTP}}$ is a quantum channel. Moreover, it is a random purification channel, in that for any $\rho\in \mathrm{D}(\calH_{A})$ of rank at most $k=\dim(\calH_\reg{B})$, we have
  \begin{align}\label{eq:random_pur_output}
    \calP^{\mathrm{CPTP}}(\rho^{\otimes n}) = \calP^{\mathrm{CP}}(\rho^{\otimes n}) = \calP(\rho^{\otimes n}) = \E \left[(\ketbra{\bm{\phi}^\rho}{\bm{\phi}^\rho})^{\otimes n}\right]
  \end{align}
  where $\ket{\bm{\phi}^\rho}\coloneq(I_{\reg{A}}\otimes \bm{V})\ket{\phi^\rho}$ for a Haar-random $\bm{V}\in\mathrm{U}(\calH_{\reg{B}})$ and $\ket{\phi^{\rho}}\in \calH_{\reg{A}}\otimes \calH_{\reg{B}}$ an arbitrary fixed purification of $\rho$.
\end{theorem}
For the special case of $\dim(\calH_\reg{B})=\dim(\calH_\reg{A})$, our description of the random purification channel coincides with that in a recent work \cite[arXiv version $2$]{girardi2026random}. \Cref{thm:rand_pur} addresses the general case of $\dim(\calH_\reg{B})\leq \dim(\calH_\reg{A})$ which the authors left as an open question.

We omit the statement for the wider class of exchangeable operators for clarity since the resulting output is not readily interpretable as a random tensor power of a purification. We also remark that if one is interested in Kraus operators for the random purification channel $\calP^{\mathrm{CPTP}}$, a natural choice is furnished by the projectors onto the eigenspaces of the Choi--Jamio\l{}kowski operator $J_\calR$. Before proving the theorem, we record the following helpful result.
\begin{lemma}[Sufficient condition for CP]\label{lem:mat_mult_cp_condn}
  Let $A\in \mathrm{L}(\calH)$ for some Hilbert space $\calH$ and $\calS$ be a subalgebra $\calS\subseteq \mathrm{L}(\calH)$. The map $f\colon \calS\to\mathrm{L}(\calH)$ with the action $f(X)=A X$ for every $X\in \calS$ is CP if $A\succeq 0$ and $[X,A]=0$ for all $X\in \calS$.
\end{lemma}
\begin{proof}
  If $A\succeq 0$ then it has a positive semidefinite square-root $\sqrt{A}$ diagonal in the same basis. Thus, for any $X\in \calS$, if $[X,A]=0$ then $[X,\sqrt{A}]=0$ as well and the action of $f$ can be written as $f(X) = \sqrt{A}X\sqrt{A}$, which is manifestly CP.
\end{proof}

\begin{proof}[Proof of \Cref{thm:rand_pur}]
  We first show that the restriction of $\calP$ to exchangeable operators is CP, which implies $\calP^{\mathrm{CP}}=\calP\circ \calT_{\mathfrak{S}_n}$ is CP on the full domain. Note that $\vee^n\mathrm{L}(\calH_\reg{A})\otimes I_{\reg{B}^n}$ is a subalgebra of $\mathrm{L}(\calH_{\reg{A}^n}\otimes \calH_{\reg{B}^n})$. Thus, by \Cref{lem:mat_mult_cp_condn}, it suffices to show that if $X\in \vee^n\mathrm{L}(\calH_{\reg{A}})$ then $[J_\calR,X\otimes I_{\reg{B}^n}]=0$. 

  By Haar invariance, $M$ and hence $\sqrt{M^+}$ commute with $U^{\otimes n}$ for every $U\in \mathrm{U}(\calH_\mathsf{A})$. Using the definition of $J_\calR$, \Cref{lem:max_ent_state_properties}, and the fact that $U^\dagger\bm{W}$ and $\bm{W}$ are identically distributed, we have
  \begin{equation*}
      \left((U^T)^{\otimes n} \otimes I_{\mathsf{B}^n}\right) J_{\calR} \left((\overline{U})^{\otimes n} \otimes I_{\mathsf{B}^n}\right) = J_{\calR} \,.
  \end{equation*}
  Since this holds for every $U^T\in \mathrm{U}(\calH_\mathsf{A})$, it means that $J_{\calR} \in \Span\{F_{\pi}^d: \pi\in \mathfrak{S}_n\} \otimes \mathrm{L}(\calH_\mathsf{B}^n)$. 
  Since $X$ commutes with any permutation operator on $\calH_{\mathsf{A}^n}$, we have that $[J_\calR,X\otimes I_{\reg{B}^n}]=0$. 

  Next, let $P$ denote the projector onto the image of $M$. We compute
  \begin{align*}
    \tr(\calP^{\mathrm{CP}}(X)) = \tr(\calR(\calT_{\mathfrak{S}_n}(X)^T)) = \tr(P \calT_{\mathfrak{S}_n}(X^T)) = \tr(P X^T) = \tr(PX).
  \end{align*}
  The first equality follows from $\calR(\cdot)=\tr_{\reg{A}^n}(\calP((\cdot)^T))$, the second from the definition of $\calR$ and the fact that transpose commutes with $\calT_{\mathfrak{S}_n}$, the third is because $\calT_{\mathfrak{S}_n}$ is self-adjoint and fixes $P$, and the final equality follows from $P^T = P$. Consequently, $\tr(\calP^{\mathrm{CPTP}}(X)) = \tr((P+Q)X) = \tr(X)$, so $\calP^{\mathrm{CPTP}}$ is CPTP.

  It remains to show \cref{eq:random_pur_output}. Since $\rho^{\otimes n}$ commutes with any permutation operator, it is an element of $\vee^n\mathrm{L}(\calH_{\reg{A}})$. Moreover, the permutation twirl $\calT_{\mathfrak{S}_n}$ acts trivially on this subspace, so $\calP(\rho^{\otimes n}) = \calP^\mathrm{CP}(\rho^{\otimes n}) = \calP^{\mathrm{CPTP}}(\rho^{\otimes n})$, where the second equality is by \Cref{lem:supp_of_c_op}. Thus, it suffices to show the final equality in \cref{eq:random_pur_output}. 
  
  To this end, write the fixed purification as $|\phi^\rho\rangle=(Y\otimes I_{\mathsf{B}})|\Gamma\rangle_{\mathsf{BB}}$ for some $Y:\mathcal{H}_{\mathsf{B}}\to\mathcal{H}_{\mathsf{A}}$, so that $\rho=YY^\dagger$. Applying \cref{eq:rand_joint_dim_red} and \Cref{lem:supp_of_c_op} to $W_0$ and using \Cref{lem:max_ent_state_properties}, we have
  \begin{equation}\label{eq:commute_choi}
  \bigl((W_0^\dag)^{\otimes n}\otimes I_{\reg{B}^n}\bigr)J_\calR\bigl(W_0^{\otimes n}\otimes I_{\reg{B}^n}\bigr)
  =\E_{\bm{V}}\left[\bigl(I_{\reg{B}^n}\otimes\bm{V}^{\otimes n}\bigr)\Gamma_{\reg{BB}}^{\otimes n}\bigl(I_{\reg{B}^n}\otimes(\bm{V}^\dag)^{\otimes n}\bigr)\right],
  \end{equation}
  where $\bm{V}\in\mathrm{U}(\calH_{\reg{B}})$ is Haar-random.
  Therefore,
\begin{align*}
  \E\left[
    (\ketbra{\bm{\phi}^\rho}{\bm{\phi}^\rho})^{\otimes n}
  \right]
  &=
  \E\left[
    \bigl((Y\otimes\bm{V})\Gamma_{\reg{BB}}
    (Y^\dag\otimes\bm{V}^\dag)\bigr)^{\otimes n}
  \right]\\
  &=
  \bigl((YW_0^\dag)^{\otimes n}\otimes I_{\reg{B}^n}\bigr)
  J_\calR
  \bigl((W_0Y^\dag)^{\otimes n}\otimes I_{\reg{B}^n}\bigr) \tag{\cref{eq:commute_choi}}\\
  &=
  J_\calR
  \bigl((YW_0^\dag W_0Y^\dag)^{\otimes n}
        \otimes I_{\reg{B}^n}\bigr) \\
  &=
  \calP(\rho^{\otimes n}),
\end{align*}
where the third equality is because $[J_\calR, X \otimes I_{\mathsf{B}^n}] = 0$ for any exchangeable $X$.
\end{proof}

\subsection{Optimal quantum tomography without representation theory}\label{sec:tomography}
We can now give an elementary description of a sample-optimal protocol for mixed state tomography of rank-$r$ quantum states. It is known from prior work~\cite{haah2016sample,odonnell2016efficient,odonnell2017efficientii,pelecanos2025mixedstatetomographyreduces} that the sample complexity of this task is $\Theta(dr/\delta)$ where $d$ is the dimension and $\delta$ is the accuracy in infidelity (one minus the fidelity). The lower bound comes from information-theoretic arguments, and the upper bound is obtained by analyzing algorithms using Schur polynomials. An optimal bound in terms of the success probability as well was recently provided in Ref.~\cite{pelecanos2025mixedstatetomographyreduces}.

The argument in this section recovers the tight upper bound without any reference to the Schur transform, being essentially basis-free. This means in particular that there is no hope of characterizing the gate complexity in this way. Nevertheless, the description suffices to recover the optimal sample complexity in a succinct manner.

The approach is identical to Ref.~\cite{pelecanos2025mixedstatetomographyreduces}: Given $n$ copies of an unknown rank-$r$ density matrix $\rho\in \mathrm{D}(\C^d)$, apply the random purification channel $\calP^{\mathrm{CPTP}}$ defined in \cref{eq:tp_rand_pur} to $\rho^{\otimes n}$, with $k$ set to $r$. This yields $n$ identical copies of a random purification $\ket{\bm{\phi}^{\rho}}\in \C^d\otimes \C^r$ of $\rho$ which satisfies $\tr_2(\ketbra{\bm{\phi}^{\rho}}{\bm{\phi}^\rho})=\rho$ with certainty. Recall that the action of $\calP^{\mathrm{CPTP}}$ on exchangeable operators can be derived from an elementary analysis of random isometries applied to entangled states, as in \Cref{sec:connection_description}. Therefore, it remains only to show that tomography of $d$-dimensional \emph{pure} states to within $\delta$ accuracy in infidelity can be accomplished using $n=O(d/\delta)$ copies, and then set $d\to dr$.

For completeness we reproduce this well-known result, which has appeared in multiple prior works~\cite{hayashi1998asymptotic,harrow2013churchsymmetricsubspace,pelecanos2025mixedstatetomographyreduces}. We let $\Pi^n_{\textnormal{sym}}\colon (\C^d)^{\otimes n}\to \vee^n(\C^d)$ denote the projector onto the symmetric subspace of $(\C^d)^{\otimes n}$, and set $d[n]\coloneq \binom{d+n-1}{n}$. We also let $\mathrm{d}u$ denote the Haar measure on the unit sphere in $\C^d$.
\begin{theorem}[Sample-optimal pure state tomography]
  Let $\ket{\psi}\in \C^d$ be an unknown pure state. With $n=O(d/\delta)$, measuring $\ket{\psi}^{\otimes n}$ using the POVM $\{d[n]\ketbra{u}{u}^{\otimes n}\mathrm{d} u\}_u$ gives an outcome $\bm{u}$ such that, with probability at least $9/10$, the fidelity between $\ket{\bm{u}}$ and $\ket{\psi}$ is at least $1-O(\delta)$.
\end{theorem}
\begin{proof}
Let $\ket{\bm{v}}\in \C^d$ denote a Haar-random state and compute
\begin{align}
  \E |\langle \bm{u}|\psi\rangle|^2&=\int \tr(d[n]\ketbra{u}{u}^{\otimes n}\cdot \ketbra{\psi}{\psi}^{\otimes n}) |\langle u|\psi\rangle|^2 \mathrm{d}u\nonumber\\
  &= d[n] \tr\left(\E\ketbra{\bm{v}}{\bm{v}}^{\otimes n+1}\cdot \ketbra{\psi}{\psi}^{\otimes n+1}\right)\nonumber\\
  &= \frac{d[n]}{d[n+1]}\tr\left(\Pi^{n+1}_{\textnormal{sym}}\cdot \ketbra{\psi}{\psi}^{\otimes n+1}\right)\nonumber\\
  &= \frac{d[n]}{d[n+1]}\label{eq:pure_state_tomog_bound}
\end{align}
where the third line follows from the identity $\E\ketbra{\bm{v}}{\bm{v}}^{\otimes n}=\Pi^n_{\textnormal{sym}}/d[n]$ (see, e.g., \cite[Proposition~6]{harrow2013churchsymmetricsubspace}) and the final line follows since $\ket{\psi}^{\otimes n+1}\in \vee^{n+1}(\C^d)$. With $n=O(d/\delta)$, the right-hand side of \cref{eq:pure_state_tomog_bound} is at least $1-O(\delta)$.
This completes the proof using Markov's inequality.
\end{proof}
The improvement of this statement to a high probability bound is given as Proposition~5.1 in~\cite{pelecanos2025mixedstatetomographyreduces}.

\subsection{Impossibility of non-exchangeable random purification}\label{sec:nogo}
For operators of the form $\rho^{\otimes n}$ where $\rho\in\mathrm{D}(\C^d)$ has rank at most $k$, the action of the dimension reduction map can be recovered by first applying the random purification channel and then tracing out the data, leaving behind the purifying registers. (The transpose does not affect the output in this case.) On the other hand, as we have seen, dimension reduction can be applied to a broader class of states, including those of the form $\rho^{\otimes a}\otimes \sigma^{\otimes b}$ for $\rho\neq \sigma$.

In this section, we show that, by contrast, states of this form cannot be randomly purified; at least, by one natural definition of the action one might desire from such a non-exchangeable random purification. This rules out the interpretation of dimension reduction as simply discarding purifying registers in a random purification, in all generality.
\begin{theorem}[No random joint purification channel]\label{thm:nogo}
    Let $d\ge 2$, $n\ge 1$, and $2\le k\le d$. Define
    \begin{equation*}
        \mathcal H_{\data}:=(\C^d)^{\otimes n}\otimes (\C^d)^{\otimes n} \,,
        \qquad
        \mathcal H_{\purify}:=(\C^k)^{\otimes n}\otimes (\C^k)^{\otimes n} \,.
    \end{equation*}
    Then there does not exist a quantum channel $\calP_{d,k,n}: \mathrm{D}(\calH_{\data}) \to \mathrm{D}(\calH_{\data} \otimes \calH_{\purify})$ with the following property:  
    for any $\rho, \sigma \in \mathrm{D}(\C^d)$ and $\widetilde{\rho},\widetilde{\sigma}\in \mathrm{D}(\C^k)$ such that $W\widetilde{\rho} \, W^\dagger = \rho$ and $W\widetilde{\sigma} W^\dagger = \sigma$ for some isometry $W:\C^k\to \C^d$, there is a probability measure $\nu$ on $\mathrm{U}(k)$ such that
    \begin{equation*}
        \Tr_{\purify}\paren*{\calP_{d,k,n}(\rho^{\otimes n} \otimes \sigma^{\otimes n})} = \rho^{\otimes n} \otimes \sigma^{\otimes n} \,,
    \end{equation*}
    and
    \begin{equation}\label{eq:marginal_purify}
        \Tr_{\data}\paren*{\calP_{d,k,n}(\rho^{\otimes n} \otimes \sigma^{\otimes n})} = \mathbb{E}_{\bU\sim \nu} \left[ (\bU \widetilde{\rho} \, \bU^\dagger)^{\otimes n} \otimes (\bU \widetilde{\sigma} \, \bU^\dagger)^{\otimes n}\right] \,.
    \end{equation}
\end{theorem}
\begin{proof}
  Assume for contradiction that such a channel $\calP_{d,k,n}$ exists. We first construct two non-orthogonal pure input states $\ket{\Psi_1}, \ket{\Psi_2}\in (\C^d)^{\otimes 2n}$ whose reduced states on the purifying register of the two outputs are equal.
    
    Fix orthonormal vectors $\ket{0},\ket{1}\in \C^k$, and let $\ket{+} = \frac{1}{\sqrt{2}} (\ket{0} + \ket{1})$. 
    Choose any isometry $W: \C^k \to \C^d$, and define
    \begin{equation*}
        \widetilde{\rho}_1 \coloneqq \ketbra{0}{0} \,, \qquad \widetilde{\sigma}_1 \coloneqq \ketbra{0}{0} \,, \qquad \widetilde{\rho}_2 \coloneqq \ketbra{0}{0} \,, \qquad \widetilde{\sigma}_2 \coloneqq \ketbra{+}{+} \,.
    \end{equation*}
    For each $i=1,2$, define $\rho_i \coloneqq W \widetilde{\rho}_i W^\dagger$ and $\sigma_i \coloneqq W \widetilde{\sigma}_i W^\dagger$. 
    We now have two pure input states
    \begin{equation*}
        \Psi_i = \ketbra{\Psi_i}{\Psi_i} = \rho_i^{\otimes n} \otimes \sigma_i^{\otimes n}
    \end{equation*}
    to the quantum channel $\calP_{d,k,n}$ for which $\braket{\Psi_1|\Psi_2} = (\braket{0|0}\cdot\braket{0|+})^n \neq 0$. 

    Let $V: \calH_{\data} \to \calH_{\data} \otimes \calH_{\purify} \otimes \calH_{\env}$ be a Stinespring isometry for $\calP_{d,k,n}$. 
    Then 
    \begin{equation*}
        \Tr_{\env,\purify}\paren{V \proj{\Psi_i} V^\dagger} = \Tr_{\purify}\paren*{\calP_{d,k,n}(\proj{\Psi_i})} = \proj{\Psi_i} 
    \end{equation*}
    is a pure state for each $i=1,2$. So there exists a unit vector $\ket{\chi_i} \in \calH_\purify \otimes \calH_{\env}$ such that
    \begin{equation*}
        V\ket{\Psi_i} = \ket{\Psi_i} \otimes \ket{\chi_i} \,.
    \end{equation*}
    Because $V$ is an isometry, then
    \begin{equation*}
        \braket{\Psi_1|\Psi_2} = \braket{\Psi_1|V^\dagger V |\Psi_2} = \braket{\Psi_1|\Psi_2} \cdot \braket{\chi_1|\chi_2} \,.
    \end{equation*}
    Because $\braket{\Psi_1|\Psi_2} \neq 0$, then $\ket{\chi_1} = \ket{\chi_2}$. 
    Since
    \begin{equation*}
        \Tr_{\data}\paren*{\calP_{d,k,n}(\proj{\Psi_i})} = 
        \Tr_{\data,\env}\paren{V \proj{\Psi_i} V^\dagger} = \Tr_{\env}\paren{\proj{\chi_i}} \,,
    \end{equation*}
    then
    \begin{equation*}
        \Tr_{\data}\paren*{\calP_{d,k,n}(\proj{\Psi_1})} = \Tr_{\data}\paren*{\calP_{d,k,n}(\proj{\Psi_2})} \,.
    \end{equation*}
    By \cref{eq:marginal_purify}, there exist probability measures $\nu_1, \nu_2$ on $\mathrm{U}(k)$ such that 
    \begin{equation}\label{eq:nogo_same_marginal}
        \mathbb{E}_{\bU\sim \nu_1} \left[ (\bU \widetilde{\rho}_1 \bU^\dagger)^{\otimes n} \otimes (\bU \widetilde{\sigma}_1 \bU^\dagger)^{\otimes n}\right] = \mathbb{E}_{\bU\sim \nu_2} \left[ (\bU \widetilde{\rho}_2 \bU^\dagger)^{\otimes n} \otimes (\bU \widetilde{\sigma}_2 \bU^\dagger)^{\otimes n}\right]
    \end{equation}
    
    We now show that this is impossible. 
    Write $\calH_\purify = A^n B^n$, and let $S \coloneqq \bigotimes_{j=1}^n \SWAP_{A_jB_j}$. Then for any states $\tau, \eta\in \mathrm{D}(\C^k)$, we have that
    \begin{equation*}
        \Tr \paren*{S \cdot \tau^{\otimes n} \otimes \eta^{\otimes n}} = \paren*{\Tr(\tau \eta)}^n \,.
    \end{equation*}
    Therefore, for any probability measure $\nu$ on $\mathrm{U}(k)$, we have that
    \begin{equation*}
        \Tr \paren*{S \cdot \E_{\bU\sim \nu} \left[ (\bU\tau \bU^\dagger)^{\otimes n} \otimes (\bU\eta \bU^\dagger)^{\otimes n} \right]} = \E_{\bU\sim \nu} \paren*{\Tr(\tau \eta)}^n = \paren*{\Tr(\tau \eta)}^n \,,
    \end{equation*}
    which is independent of $\nu$. By \cref{eq:nogo_same_marginal}, we must have $(\Tr(\widetilde{\rho}_1 \widetilde{\sigma}_1))^n = (\Tr(\widetilde{\rho}_2 \widetilde{\sigma}_2))^n$. However this is false because $\Tr(\widetilde{\rho}_1 \widetilde{\sigma}_1) = 1$ and $\Tr(\widetilde{\rho}_2 \widetilde{\sigma}_2) = \frac12$. This completes our proof by contradiction. 
\end{proof}

\section{Application to learning symmetric properties}
\label{sec:continuity}

In this section, we prove the improved sample complexity upper bounds stated in \Cref{tab:bounds}. The strategy is as follows: we first recall the sample-optimal full state tomography guarantees in fidelity (\Cref{thm:tomography_samples}), then establish conversion lemmas (\Cref{lem:conversion}) that translate tomography accuracy into estimation accuracy for symmetric properties such as the trace distance, fidelity, and quantum Tsallis relative entropy. The proof of the bounds in \Cref{tab:bounds} then follows by combining these conversion lemmas with the random dimension reduction result (\Cref{cor:rho_sigma_cor}).

Denote by $\norm{A}_p$ the Schatten-$p$ norm of a matrix $A$. 
\begin{definition}[A collection of symmetric properties]
    For any $\rho, \sigma\in \mathrm{D}(\C^d)$, we recall the definitions of the following symmetric properties between $\rho$ and $\sigma$: 
    \begin{itemize}
        \item (Trace distance) $D_{\mathrm{tr}}(\rho, \sigma) = \frac12 \norm{\rho - \sigma}_1$
        \item (Fidelity) $F(\rho, \sigma) = \norm{\sqrt{\rho}\sqrt{\sigma}}_1^2$. We also use the square root fidelity $\sqrt{F}(\rho, \sigma) = \sqrt{F(\rho, \sigma)}$. 
        \item (Quantum Hellinger distance) $D_{\mathrm{H}}(\rho, \sigma) = \sqrt{2(1 - \tr(\sqrt{\rho}\sqrt{\sigma}))}$
        \item (Quantum $\alpha$-Tsallis relative entropy) $D^{\mathrm{Tsa}}_{\alpha}(\rho \Vert \sigma) = \frac{1}{1-\alpha}(1 - \tr(\rho^{\alpha} \sigma^{1-\alpha}))$
    \end{itemize}
\end{definition}

\begin{theorem}[State tomography in fidelity, {\cite[Theorem 1.3]{pelecanos2025mixedstatetomographyreduces}}]\label{thm:tomography_samples}
For any unknown state $\rho \in \mathrm{D}(\C^d)$, there is an algorithm that outputs a classical description $\widehat{\rho}$ with probability $99\%$ such that $F(\rho, \widehat{\rho} \, ) \geq 1- \delta$ using $O(d^2/\delta)$ samples of $\rho$.
\end{theorem}

Before proving the continuity bounds in \Cref{lem:conversion}, we record an auxiliary estimate that converts fidelity guarantees into control of the overlap $\tr(\rho^\alpha\sigma^{1-\alpha})$, which is the quantity appearing in the quantum $\alpha$-Tsallis relative entropy.

\begin{lemma}\label{lem:alpha-overlap-fidelity-continuity}
Let $0<\alpha\leq \frac12$. For any quantum states $\rho,\widehat{\rho},\sigma,\widehat{\sigma}$, we have
\begin{equation*}
\abs*{
    \tr(\rho^\alpha\sigma^{1-\alpha})
    -
    \tr(\widehat{\rho}^{\,\alpha}
        \widehat{\sigma}^{\,1-\alpha})
}\leq 2^{\alpha}(1-F(\rho, \widehat{\rho} \, ))^{\alpha} + 3\sqrt{2} \sqrt{1 - F(\sigma, \widehat{\sigma})} \,.
\end{equation*}
\end{lemma}
\begin{proof}
By adding and subtracting
$\tr(\widehat{\rho}^{\,\alpha}\sigma^{1-\alpha})$
and using the triangle inequality, we obtain
\begin{align*}
\abs*{\tr(\rho^\alpha\sigma^{1-\alpha}) - \tr(\widehat{\rho}^{\,\alpha}
    \widehat{\sigma}^{\,1-\alpha})}
\leq
\abs*{\tr\paren*{
(\rho^\alpha-\widehat{\rho}^{\,\alpha})
\sigma^{1-\alpha}}}
+
\abs*{\tr\paren*{
\widehat{\rho}^{\,\alpha}
(\sigma^{1-\alpha}
-\widehat{\sigma}^{\,1-\alpha})}} \,.
\end{align*}
By H\"older's inequality for Schatten norms,
\begin{align*}
\abs*{
\tr\paren*{
(\rho^\alpha-\widehat{\rho}^{\,\alpha})
\sigma^{1-\alpha}}
}
&\leq
\norm*{\rho^\alpha-\widehat{\rho}^{\,\alpha}}_{1/\alpha}
\norm*{\sigma^{1-\alpha}}_{1/(1-\alpha)}
\\
&=
\norm*{\rho^\alpha-\widehat{\rho}^{\,\alpha}}_{1/\alpha}
\tag{$\norm*{\sigma^p}_{1/p} = \tr(\sigma)=1$} \,. 
\end{align*}
Similarly, we have
\begin{align*}
\abs*{\tr\paren*{\widehat{\rho}^{\,\alpha}
(\sigma^{1-\alpha}
-\widehat{\sigma}^{\,1-\alpha})}} \leq
\norm*{\sigma^{1-\alpha}-\widehat{\sigma}^{1-\alpha}}_{1/(1-\alpha)} \,.
\end{align*}
We use the following standard power estimates for positive semidefinite operators, obtained from
\cite[Lemmas~2.1 and~2.2]{Ricard2015}:
\begin{align}
\norm{X^r-Y^r}_{2/r}
&\leq
\norm{X-Y}_2^r,
&&0<r\leq1,
\label{eq:small-power}
\\
\norm{X^r-Y^r}_{2/r}
&\leq
3\norm{X-Y}_2 \cdot 
\max\!\left\{
    \norm{X}_2,\norm{Y}_2
\right\}^{r-1},
&&1<r<2.
\label{eq:large-power}
\end{align}
Then using \cref{eq:small-power} with $r=2\alpha$, $X = \sqrt{\rho}$, and $Y = \sqrt{\widehat{\rho} \,}$, we have
\begin{equation*}
    \norm*{\rho^\alpha-\widehat{\rho}^{\,\alpha}}_{1/\alpha} \leq \norm[\big]{\sqrt{\rho} - \sqrt{\widehat{\rho} \,} }_2^{2\alpha} \,.
\end{equation*}
Using \cref{eq:large-power} with $r=2(1-\alpha)$, $X = \sqrt{\sigma}$, and $Y = \sqrt{\widehat{\sigma} \,}$, we have
\begin{equation*}
    \norm*{\sigma^{1-\alpha}-\widehat{\sigma}^{1-\alpha}}_{1/(1-\alpha)} \leq 3 \norm[\big]{\sqrt{\sigma} - \sqrt{\widehat{\sigma}}}_2 \,.
\end{equation*}
We now bound $\norm{\sqrt{\rho} - \sqrt{\widehat{\rho} \,} }_2$ and $\norm{\sqrt{\sigma} - \sqrt{\widehat{\sigma}} }_2$.
For any quantum states $\tau, \omega$, 
\begin{align*}
\sqrt{F(\tau,\omega)}
&=
\norm{\sqrt{\tau}\sqrt{\omega}}_1
\\
&=
\norm{
    \tau^{1/4} \cdot 
    \bigl(\tau^{1/4}\omega^{1/4}\bigr) \cdot 
    \omega^{1/4}
}_1
\\
&\leq
\norm{\tau^{1/4}}_4 \cdot 
\norm{\tau^{1/4}\omega^{1/4}}_2 \cdot
\norm{\omega^{1/4}}_4 \tag{H\"older's inequality for Schatten norms}
\\
&=
\sqrt{\tr(\sqrt{\tau}\sqrt{\omega})}.
\end{align*}
Consequently, 
\begin{equation*}
\norm{\sqrt{\tau}-\sqrt{\omega}}_2^2 = \tr\!\left((\sqrt{\tau}-\sqrt{\omega})^2\right) = 2-2\tr(\sqrt{\tau}\sqrt{\omega}) \leq
2\bigl(1-F(\tau,\omega)\bigr).
\end{equation*}
Plugging them back in completes the proof. 
\end{proof}

We now collect the fidelity-to-property continuity bounds used in the applications below.

\begin{lemma}[Continuity bounds for symmetric properties]\label{lem:conversion}
    Let $0\leq \delta \leq \frac12$. For any $\rho, \widehat{\rho}, \sigma, \widehat{\sigma} \in \mathrm{D}(\C^d)$ such that $F(\rho, \widehat{\rho} \, )\geq 1 - \delta$ and $F(\sigma, \widehat{\sigma} \, )\geq 1 - \delta$, we have
    \begin{enumerate}
        \item\label{eq:convert_trace}
            $\abs*{D_{\tr}(\rho,\sigma) - D_{\tr}(\widehat{\rho}, \widehat{\sigma})} \leq  2\sqrt{\delta}$ ,
        \item\label{eq:convert_fidelity}
        $\frac12\abs*{F(\rho, \sigma) - F(\widehat{\rho}, \widehat{\sigma})} \leq \abs[\big]{\sqrt{F}(\rho, \sigma) - \sqrt{F}(\widehat{\rho}, \widehat{\sigma})} \leq 2\sqrt{\delta}$ , 
        \item\label{eq:convert_tsa}
        $\abs*{D_{\alpha}^{\mathrm{Tsa}}(\rho\Vert \sigma) - D_{\alpha}^{\mathrm{Tsa}}(\widehat{\rho} \, \Vert \widehat{\sigma})} \leq \frac{4(2\delta)^{\alpha}}{1-\alpha}$ for any $0< \alpha \leq 1/2$ .
    \end{enumerate}
\end{lemma}
\begin{proof}
    \cref{eq:convert_trace} follows from the triangle inequality for the trace distance:
    \begin{equation*}
        \abs*{D_{\tr}(\rho,\sigma) - D_{\tr}(\widehat{\rho}, \widehat{\sigma})} \leq D_{\tr}(\rho,\widehat{\rho} \, ) + D_{\tr}(\sigma, \widehat{\sigma})\leq \sqrt{1-F(\rho, \widehat{\rho} \, )} + \sqrt{1-F(\sigma, \widehat{\sigma} )} \leq 2\sqrt{\delta} \,.
    \end{equation*}
    \cref{eq:convert_tsa} follows from \Cref{lem:alpha-overlap-fidelity-continuity}:
    \begin{align*}
        \abs*{D_{\alpha}^{\mathrm{Tsa}}(\rho\Vert \sigma) - D_{\alpha}^{\mathrm{Tsa}}(\widehat{\rho} \, \Vert \widehat{\sigma})} &\leq \frac{1}{1-\alpha} \abs*{
        \tr(\rho^\alpha\sigma^{1-\alpha})
        -
        \tr(\widehat{\rho}^{\,\alpha}
            \widehat{\sigma}^{\,1-\alpha})
        } \\
        &\leq \frac{1}{1-\alpha} \paren*{ 2^{\alpha}(1-F(\rho, \widehat{\rho} \, ))^{\alpha} + 3\sqrt{2} \sqrt{1 - F(\sigma, \widehat{\sigma})} }\tag{\Cref{lem:alpha-overlap-fidelity-continuity}}\\
        &\leq \frac{1}{1-\alpha}\paren*{(2\delta)^{\alpha} + 3\sqrt{2\delta}} \\
        &\leq \frac{4}{1-\alpha}(2\delta)^{\alpha} \,. \tag{$0\leq \delta \leq 1/2$ and $0<\alpha\leq 1/2$}
    \end{align*}
    The rest of this proof is to show \cref{eq:convert_fidelity}. 
    Since $|a^2-b^2| = |a-b|\cdot(a+b)\leq 2|a-b|$ for all $0\leq a,b \leq 1$, it suffices to show the second inequality.

    Note that the Bures angle defined as 
    \begin{equation*}
        \theta(\rho, \sigma) \coloneqq \arccos \paren[\big]{\sqrt{F}(\rho, \sigma)}
    \end{equation*}
    is a valid distance metric on the space of quantum states \cite{Gilchrist_2005}, and hence it satisfies the triangle inequality:
    \begin{equation*}
        \abs*{\theta(\rho, \sigma) - \theta(\widehat{\rho}, \widehat{\sigma})} \leq \theta(\rho, \widehat{\rho} \, ) + \theta(\sigma, \widehat{\sigma}) \,.
    \end{equation*}
    Let $\gamma = \arccos(\sqrt{1-\delta})$. Then $\cos(\gamma) = \sqrt{1-\delta}$ and $\sin(\gamma) = \sqrt{\delta}$. 
    Since the arccosine function is monotonically decreasing, then $\theta(\rho, \widehat{\rho} \, ) \leq \gamma$ and $\theta(\sigma, \widehat{\sigma} ) \leq \gamma$. 
    Let $\theta_1 = \theta(\rho, \sigma)$ and $\theta_2 = \theta(\widehat{\rho}, \widehat{\sigma})$. 
    Then 
    \begin{equation*}
        \abs[\big]{\sqrt{F}(\rho, \sigma) - \sqrt{F}(\widehat{\rho}, \widehat{\sigma})} = 
        \abs*{\cos(\theta_1) - \cos(\theta_2)} \,.
    \end{equation*}
    Hence the goal is to maximize $\abs*{\cos(\theta_1) - \cos(\theta_2)}$ subject to $\abs{\theta_1 - \theta_2}\leq 2\gamma$. Suppose $\theta_1\leq \theta_2$. Then the maximum is achieved at $\theta_1 =\frac{\pi}{2} - 2\gamma$ and $\theta_2 =\frac{\pi}{2}$. So
    \begin{align*}
    \abs[\big]{\sqrt{F}(\rho, \sigma) - \sqrt{F}(\widehat{\rho}, \widehat{\sigma})} \leq \sin(2\gamma) = 2\sin(\gamma)\cos(\gamma) = 2\sqrt{\delta} \sqrt{1-\delta} \leq 2\sqrt{\delta} \,. & \qedhere
    \end{align*}
\end{proof}

\begin{proof}[Proof of the bounds in \Cref{tab:bounds}]
  Let $f(\rho, \sigma)$ be a symmetric property between $\rho$ and $\sigma$ and $r$ be the maximum of the ranks of $\rho$ and $\sigma$.
By \Cref{cor:rho_sigma_cor}, when applying random dimension reduction to $\rho^{\otimes n} \otimes \sigma^{\otimes n}$, we obtain a dimension-reduced random state $(\bU \rho_0 \bU^\dagger)^{\otimes n} \otimes (\bU \sigma_0 \bU^\dagger)^{\otimes n}$ where $\bU$ is a $2r$-dimensional Haar-random unitary, $\bU \rho_0 \bU^\dagger$ and $\bU \sigma_0 \bU^\dagger$ are $2r$-dimensional, and the symmetric property is preserved: 
\begin{equation*}
    f(\rho, \sigma) = f(\bU \rho_0 \bU^\dagger, \bU \sigma_0 \bU^\dagger) \,.
\end{equation*}
Then the sample complexity upper bounds for the four estimation tasks listed in \Cref{tab:bounds} follow from \Cref{lem:conversion} and performing the corresponding full state tomography of $\bU \rho_0 \bU^\dagger$ and $\bU \sigma_0 \bU^\dagger$ using the fidelity guarantees in \Cref{thm:tomography_samples}. 
\end{proof}

\section{Efficient implementation from the Schur transform}\label{sec:schur_and_efficient}
\subsection{Representation theory basics}
In this section, we set up the necessary representation theory ingredients and notations. (See the PhD theses~\cite{Harrow05thesis,wright2016}, or \cite{fulton2013representation,goodman2009symmetry} for standard introductions to these topics.) 
\begin{itemize}
    \item The irreducible representations of the symmetric group $\mathfrak{S}_n$ are indexed by partitions $\lambda\vdash n$ and are written $(\kappa_\lambda, \Sp_\lambda)$. 

    \item The irreducible polynomial representations of the general linear group $\GL(d)$ are indexed by partitions $\lambda$ in which $\ell(\lambda)\leq d$ and are written $(\nu_\lambda^d, V_\lambda^d)$. 
    We also use $\nu_\lambda^d$ as the polynomial extension to evaluate any matrix $M\in \mathrm{L}(\C^d)$ not necessarily invertible.     
\end{itemize}
One may view $\calH_{\reg{A}^n} = (\C^d)^{\otimes n}$ as a representation space for the group $\mathfrak{S}_n\times \mathrm{U}(d)$ via $(\pi, U) \mapsto F_{\pi}^d U^{\otimes n}$. 
Since this representation is unitary, it has a canonical decomposition into a direct sum of orthogonal, isotypic components $(\C^d)^{\otimes n} = \bigoplus_{\lambda}\calV_\lambda$. Here and throughout we let $\Pi_\lambda$ denote the projection onto the $\lambda$th subspace in this direct sum. By Schur--Weyl duality, the labels for the direct sum are partitions $\lambda\vdash n$ of length at most $\ell(\lambda)\leq d$, and the corresponding isotypic component is irreducible and isomorphic to $\Sp_\lambda\otimes V^d_\lambda$. The Schur transform is the unitary transformation that intertwines these isomorphic representation spaces, i.e.,
\begin{align}
  U_\schur^{d,n} F^d_\pi M^{\otimes n}(U^{d,n}_\schur)^\dag = \bigoplus_{\lambda\vdash n\colon \ell(\lambda)\leq d} \kappa_\lambda(\pi)\otimes \nu^d_\lambda(M)\quad \textnormal{for all}\ \pi\in \mathfrak{S}_n,\ M\in \mathrm{L}(\C^d) \,.\label{eq:schur_action}
\end{align}
It is well-known that the Schur transform can be implemented efficiently: 
\begin{theorem}[Efficient implementation of the Schur transform {\cite{burchardt2025highdimensionalquantumschurtransforms}}]\label{thm:eff_schur}
The Schur transform $U_{\schur}^{d,n}$ can be implemented to error $\epsilon$ in diamond norm with gate complexity $\poly(n, \log(d), \log(1/\epsilon))$. 
\end{theorem}
We sometimes omit the superscript $n$ and simply write $U_{\schur}^d$ when the context is clear.

\subsection{Algorithm description}\label{sec:efficient}
The change of basis in \cref{eq:schur_action}, known as the Schur basis, elucidates the actions of the objects introduced in \Cref{sec:dim_red}. For example, the following natural algorithm implements the dimension reduction channel $\calR^{\mathrm{TP}}$, which is stated explicitly in \Cref{thm:implementation} below.

\begin{tcolorbox}[colback=white, arc=0pt, boxrule=0.5pt]
  \begin{algorithm}[Random dimension reduction]\label{alg:dim_red}\;\vspace{0.5em}
  \newline
  \textbf{Input:} Quantum state $X\in \mathrm{D}((\C^d)^{\otimes n})$ such that $X = W^{\otimes n} X_0 (W^\dag)^{\otimes n}$ for some $X_0\in \mathrm{D}((\C^k)^{\otimes n})$ and isometry $W\colon \C^k\to \C^d$.\vspace{0.5em}\newline
  \textbf{Output:} $\calR^{\mathrm{TP}}(X)=\E_{\bm{U}} \bm{U}^{\otimes n} X_0 (\bm{U}^\dag)^{\otimes n}$ where $\bm{U}\in \mathrm{U}(k)$ is Haar-random.
  \begin{enumerate}
    \item Apply the Schur transform $U_{\schur}^{d,n}$ to $X$.
    \item Measure the label $\lambda\vdash n$ (weak Schur sampling). If $\ell(\lambda)>k$, output the maximally mixed state (pick a random basis state) and terminate. 
    \item If $\ell(\lambda)\leq k$, preserve the $\Sp_\lambda$ register, discard the $V_\lambda^d$ register, and initialize a fresh $V_\lambda^k$ register in the maximally mixed state.
    \item Apply the inverse Schur transform $(U_{\schur}^{k,n})^\dagger$.
  \end{enumerate}
  \end{algorithm}
\end{tcolorbox}

It follows from \Cref{thm:eff_schur} that this entire algorithm can be implemented to accuracy $\epsilon$ in diamond distance using a quantum circuit of size polynomial in $n$, $\log d$, and $\log(1/\epsilon)$. Hence, the proof of \Cref{thm:rand_joint_dim_red} is complete upon establishing the following result. Here and throughout, we let indexing by partitions denote the corresponding Schur block, e.g., if $X\in\mathrm{L}((\C^d)^{\otimes n})$ then $X_{\lambda\mu}\colon \Sp_\mu\otimes V^d_\mu\to \Sp_\lambda\otimes V^d_\lambda$ is the $(\lambda,\mu)$th block of $U^d_\schur X (U^d_\schur)^\dag$, according to the vector space decomposition in the right-hand side of \cref{eq:schur_action}. We also collect subspaces through $\Pi_{\leq k}\coloneq \sum_{\lambda\vdash n\colon \ell(\lambda)\leq k}\Pi_\lambda$ and $\Pi_{> k}\coloneq \sum_{\lambda\vdash n\colon \ell(\lambda) > k}\Pi_\lambda$.
\begin{theorem}[Dimension reduction implementation]\label{thm:implementation}
  \Cref{alg:dim_red} implements the dimension reduction channel $\calR^{\mathrm{TP}}$ defined in \cref{eq:red_tp}. Namely, for every $X\in \mathrm{L}((\C^d)^{\otimes n})$, one has
  \begin{align*}
    U_\schur^k \calR^{\mathrm{TP}}(X)(U_\schur^k)^\dag &= \tr(\Pi_{>k} X)\frac{I}{k^n} + \bigoplus_{\lambda\vdash n\colon \ell(\lambda)\leq k}\tr_{V^d_\lambda}(X_{\lambda\lambda})\otimes \frac{I_{V^k_\lambda}}{\dim(V_\lambda^k)} \,.
  \end{align*}
\end{theorem}
Before proving the theorem, we make use of the following observations. Each of these is based on Schur's lemma: an operator which commutes with a given group action block-diagonalizes in an orthonormal basis partially labelled by irreducible representations. Namely, it acts within each isotypic space and non-trivially only on the multiplicity space. This also generalizes to non-square intertwining operators, and we record these facts explicitly as \Cref{lem:commuting_implies_blocks} and \Cref{lem:intertwiner_implies_blocks}.
\begin{lemma}[Unitary twirl in the Schur basis]\label{lem:twirl_schur}
  The action of the $n$-fold unitary twirl in the Schur basis is given by
  \begin{align*}
    U^d_\schur\left[\E_{\bm{U}}\bm{U}^{\otimes n} X(\bU^\dag)^{\otimes n}\right](U^d_\schur)^\dag &= \bigoplus_{\lambda\vdash n\colon \ell(\lambda)\leq d} \tr_{V^d_\lambda}(X_{\lambda\lambda})\otimes \frac{I_{V^d_\lambda}}{\dim(V_\lambda^d)}
  \end{align*}
  for every $X\in\mathrm{L}((\C^d)^{\otimes n})$.
\end{lemma}
\begin{proof}
  Set $A\equiv \E_{\bm{U}}\bm{U}^{\otimes n} X(\bU^\dag)^{\otimes n}$. By left-invariance of the Haar measure, we have that $[A,U^{\otimes n}]=0$ for every $U\in \mathrm{U}(d)$. Hence, by Schur's lemma --- specifically, \Cref{lem:commuting_implies_blocks} --- we must have $U^d_\schur A(U^d_\schur)^\dag = \bigoplus_{\lambda\vdash n\colon \ell(\lambda)\leq d} B_\lambda \otimes I_{V^d_\lambda}$ for some $B_\lambda\in \mathrm{L}(\Sp_\lambda)$ which satisfy 
  \begin{equation*}
        B_\lambda \dim(V^d_\lambda) = \tr_{V_\lambda^d}(A_{\lambda\lambda}) = 
      \E_{\bU}\tr_{V^d_\lambda}((\bU^{\otimes n}X(\bU^\dag)^{\otimes n})_{\lambda\lambda})  \,.
  \end{equation*}
  Next, observe that
  \begin{equation*}
    (U^{\otimes n} X (U^\dag)^{\otimes n})_{\lambda\lambda} = (I_{\Sp_\lambda}\otimes \nu^d_\lambda(U))X_{\lambda\lambda}(I_{\Sp_\lambda}\otimes \nu^d_\lambda(U^\dag))
  \end{equation*}
  for any $U\in\mathrm{U}(d)$. 
  Then, by the cyclicity of partial trace in \Cref{lem:partial_trace_invariance} and the fact that $\nu^d_\lambda(U^\dag) \nu^d_\lambda(U) = I_{V_\lambda^d}$, we have
  \begin{equation*}
      \tr_{V^d_\lambda}\paren*{(I_{\Sp_\lambda}\otimes \nu^d_\lambda(U))X_{\lambda\lambda}(I_{\Sp_\lambda}\otimes \nu^d_\lambda(U^\dag))} = \tr_{V^d_\lambda}(X_{\lambda\lambda})
  \end{equation*}
  Therefore $B_\lambda \dim(V_\lambda^d) = \tr_{V^d_\lambda}(X_{\lambda\lambda})$ for each $\lambda$, which completes the proof.
\end{proof}
\begin{lemma}[Tensor power of an isometry in the Schur basis]\label{lem:isometry_schur}
For any isometry $W: \C^k \to \C^d$, the action of $W^{\otimes n}$ in the Schur basis satisfies
  \begin{align}\label{eq:isometry_schur}
    U^d_\schur W^{\otimes n} (U^k_\schur)^\dag &= \bigoplus_{\lambda\vdash n\colon \ell(\lambda)\leq k} I_{\Sp_\lambda}\otimes \eta_\lambda(W) \,,
  \end{align}
  where $\eta_\lambda(W)\colon V^k_\lambda\to V^d_\lambda$
  and $(\eta_\lambda(W))^\dag \eta_\lambda(W) = I_{V^k_\lambda}$ for every $\lambda\vdash n$ in the direct sum.
\end{lemma}
\begin{proof}
  First, note that $V^k_\lambda$ and $V^d_\lambda$ are the multiplicity spaces for $(\C^k)^{\otimes n}$ and  $(\C^d)^{\otimes n}$ viewed as $\mathfrak{S}_n$-representations, respectively. Since $W^{\otimes n}$ intertwines the symmetric group actions on these spaces, Schur's lemma --- specifically, \Cref{lem:intertwiner_implies_blocks} --- implies the left-hand side of \cref{eq:isometry_schur} must be of the form $\bigoplus_{\lambda\vdash n\colon \ell(\lambda)\leq k} I_{\Sp_\lambda}\otimes \eta_\lambda(W)$ for some $\eta_\lambda(W)\colon V^k_\lambda\to V^d_\lambda$. The fact that $(\eta_\lambda(W))^\dag \eta_\lambda(W) = I_{V^k_\lambda}$ for every $\lambda\vdash n$ follows from $W^\dag W=I$.
\end{proof}

\begin{lemma}[Grassmannian moment in the Schur basis]\label{lem:moment_schur}
Recall that $M^{(n)}_{k\embed d}\coloneq \E_{\bm{W}} (\bm{W}\bm{W}^\dag)^{\otimes n}$ where $\bm{W}\colon \C^k\to\C^d$ is a Haar-random isometry. It holds that
\begin{equation*}
    M^{(n)}_{k\embed d} = \sum_{\lambda\vdash n\colon \ell(\lambda)\leq k}\frac{\dim(V^k_\lambda)}{\dim(V^d_\lambda)}\Pi_\lambda \,.
\end{equation*}
Consequently, the projectors onto the image and kernel of $M^{(n)}_{k\embed d}$ are $\Pi_{\leq k}$ and $\Pi_{> k}$, respectively.
\end{lemma}
\begin{proof}
  The operator $M\equiv M^{(n)}_{k\embed d}$ commutes with $F^d_\pi$ for any $\pi\in \mathfrak{S}_n$ (\Cref{lem:c-commutes}) and, by left-invariance of the Haar measure, it also commutes with $U^{\otimes n}$ for any unitary operator $U\in \mathrm{U}(d)$. Hence, by Schur's lemma - specifically, \Cref{lem:commuting_implies_blocks} --- it must be a linear combination of projectors onto the isotypic components, i.e., of the form $M=\sum_{\lambda\vdash n\colon \ell(\lambda)\leq d}a_\lambda\Pi_\lambda$ for some $a_\lambda\in \C$. Using the definition of $M$, these coefficients satisfy
  \begin{align}\label{eq:lambdath_comp_of_M}
    \E_{\bU}\tr(\Pi_\lambda (\bU W_0W_0^\dag\bU^\dag)^{\otimes n}) &= a_\lambda\dim(\Sp_\lambda)\dim(V^d_\lambda)
  \end{align}
  for each $\lambda$. Note that $\Pi_\lambda$ commutes with $U^{\otimes n}$ for any unitary $U$ since it is a linear combination of permutation operators (see, e.g., \cite[Theorem~8]{Serre1977}). Hence, by the cyclic property of trace, the left-hand side of \cref{eq:lambdath_comp_of_M} above equals
  \begin{align*}
    \tr((W_0^\dag)^{\otimes n}\Pi_\lambda W_0^{\otimes n}) = \tr\left(I_{\Sp_\lambda}\otimes \eta_\lambda(W_0)^\dag\eta_\lambda(W_0)\right) = \dim(\Sp_\lambda)\dim(V^k_\lambda)
  \end{align*}
  if $\ell(\lambda)\leq k$, and is equal to zero otherwise, where we have made use of \Cref{lem:isometry_schur}. Substituting into \cref{eq:lambdath_comp_of_M} yields $a_\lambda = \dim(V^k_\lambda)/\dim(V^d_\lambda)$ for each $\lambda\vdash n$ such that $\ell(\lambda)\leq k$. 
\end{proof}

\begin{proof}[Proof of \Cref{thm:implementation}]
  Recall from \cref{eq:red_tp} that the action of $\calR^{\mathrm{TP}}$ is given by
  \begin{align*}
    \calR^{\mathrm{TP}}(X) = \tr(Q X)\frac{I}{k^n} + \E_{\bm{U}}\left[(W_0^\dag\bm{U}^\dag)^{\otimes n} \sqrt{M^+}X\sqrt{M^+}(\bm{U}W_0)^{\otimes n}\right]
  \end{align*}
  for every $X\in\mathrm{L}((\C^d)^{\otimes n})$ where $M = M^{(n)}_{k\embed d}$, $\bm{U}\in\mathrm{U}(d)$ is Haar-random, $W_0\colon \C^k\to\C^d$ is the natural embedding of $\C^k$ in $\C^d$, and $Q$ is the projector onto the kernel of $M$. By \Cref{lem:moment_schur}, we have
  $
    \sqrt{M^+} = \sum_{\lambda\colon \ell(\lambda)\leq k}\sqrt{\frac{\dim(V^d_\lambda)}{\dim(V^k_\lambda)}}\Pi_\lambda
  $
  from which we obtain
  \begin{align*}
    \sqrt{M^+}X\sqrt{M^+} &= \sum_{\lambda,\mu}a_\lambda a_\mu\Pi_\lambda X\Pi_\mu
  \end{align*}
  where $a_\lambda\coloneq \sqrt{\frac{\dim(V^d_\lambda)}{\dim(V^k_\lambda)}}$ for each $\lambda\vdash n$ such that $\ell(\lambda)\leq k$ and the sum is over all pairs of such partitions. Hence, the $(\lambda,\lambda)$th block of this matrix is equal to $a_\lambda^2 X_{\lambda\lambda}$ if $\ell(\lambda)\leq k$ and zero otherwise, from which we obtain
  \begin{align*}
    U_\schur^d\left(\E_{\bm{U}}(\bU^\dag)^{\otimes n} \sqrt{M^+}X\sqrt{M^+}\bU^{\otimes n}\right)(U_\schur^d)^\dag &= \bigoplus_{\lambda\vdash n\colon \ell(\lambda)\leq k} \tr_{V^d_\lambda}(a_\lambda^2X_{\lambda\lambda})\otimes \frac{I_{V^d_\lambda}}{\dim(V^d_\lambda)}\\
    &= \bigoplus_{\lambda\vdash n\colon \ell(\lambda)\leq k} \tr_{V^d_\lambda}(X_{\lambda\lambda})\otimes \frac{I_{V^d_\lambda}}{\dim(V^k_\lambda)}
  \end{align*}
  The proof is complete by \Cref{lem:isometry_schur} and making use of the fact that $Q=\Pi_{>k}$, by \Cref{lem:moment_schur}.
\end{proof}

\section*{Acknowledgments}

We thank Aram Harrow for helpful feedback on an early draft of this work. XT thanks John Wright for helpful discussion. We also thank Satoshi Yoshida, Akihito Soeda, and Mio Murao for bringing to our attention their earlier work~\cite{Yoshida2023universal,Yoshida2025universal} on universal decoders.

AL is supported by the NSF (PHY-2325080) and the Simons Foundation (MP-SIP-00001553, AWH).
XT is supported by the U.S. Department of Energy, Office of Science, National Quantum Information Science Research Centers, Co-design Center for Quantum Advantage (C2QA) under contract number DE-SC0012704.

\paragraph{Note added.} Simultaneous and independent work by Girardi, Anna Mele, and Lami also describes the rank-dependent random purification channel without the Schur transform, answering their own open question, in \cite{girardi2026random}, arXiv Version 3.
Another simultaneous and independent work by Liu, Tura, and Kn\"{o}rzer explores some related ideas of performing random projections to copies of quantum states in \cite{liu2026randomprojectionsmulticopyquantum}.

\printbibliography
\appendix

\section{Alternative proof of dimension reduction}
The tools established in \Cref{sec:efficient} allow us to give a second proof of \cref{eq:rand_joint_dim_red}. Let $X=W^{\otimes n} X_0 (W^\dag)^{\otimes n}$ for some $X\in\mathrm{L}((\C^d)^{\otimes n})$, $X_0\in\mathrm{L}((\C^k)^{\otimes n})$, and fixed isometry $W\colon \C^k\to\C^d$, as in \Cref{thm:rand_joint_dim_red}. By \Cref{lem:twirl_schur} we have
\begin{align*}
  U_\schur^k\left[\E_{\bU\sim \mathrm{U}(k)} \bU^{\otimes n} X_0 (\bU^\dag)^{\otimes n}\right](U_\schur^k)^\dag &= \bigoplus_{\lambda\vdash n\colon \ell(\lambda)\leq k} \tr_{V^k_\lambda}((X_0)_{\lambda\lambda})\otimes \frac{I_{V^k_\lambda}}{\dim(V^k_\lambda)}.
\end{align*}
Therefore, \cref{eq:rand_joint_dim_red} follows from \Cref{thm:implementation} upon establishing that, for such an operator $X$, (i) $\tr(\Pi_{> k} X)=0$ and (ii) $\tr_{V^d_\lambda}(X_{\lambda\lambda})=\tr_{V^k_\lambda}((X_0)_{\lambda\lambda})$. For the first item, let $W\equiv U W_0$ for some unitary $U\in \mathrm{U}(\C^d)$. Then
\begin{align*}
  \tr(\Pi_{>k}X) &= \sum_{\lambda\vdash n\colon \ell(\lambda) > k}\tr(\Pi_\lambda (UW_0)^{\otimes n} X_0 (W_0^\dag U^\dag)^{\otimes n})\\
  &= \sum_{\lambda\vdash n\colon \ell(\lambda)> k}\tr(\Pi_\lambda W_0^{\otimes n} X_0 (W_0^\dag)^{\otimes n}) \tag{$[\Pi_\lambda, U^{\otimes n}] = 0$}\\
  &= 0
\end{align*}
where the last equality is because $\Pi_\lambda W_0^{\otimes n}=0$ for any $\lambda\vdash n$ such that $\ell(\lambda)>k$, which is clear upon expressing the action of this operator in the Schur basis using \Cref{lem:isometry_schur}. Similarly, for the second item, we have
\begin{align*}
  \tr_{V^d_\lambda}(X_{\lambda\lambda}) &= \tr_{V^d_\lambda}\left([I_{\Sp_\lambda}\otimes \nu_\lambda^d(U)\eta_\lambda(W_0)](X_0)_{\lambda\lambda}[I_{\Sp_\lambda}\otimes \eta_\lambda(W_0)^\dag\nu^d_\lambda(U)^\dag]\right)\\
  &= \tr_{V^k_\lambda}\left([I_{\Sp_\lambda}\otimes \eta_\lambda(W_0)^\dag\nu_\lambda^d(U)^\dag\nu_\lambda^d(U)\eta_\lambda(W_0)](X_0)_{\lambda\lambda}\right)\\
  &= \tr_{V^k_\lambda}((X_0)_{\lambda\lambda})
\end{align*}
where the second line follows from \Cref{lem:partial_trace_invariance} and the third line follows from the fact that $\nu^d_\lambda(U)^\dag\nu^d_\lambda(U)=I$ since $\nu^d_\lambda(\cdot)$ is a unitary representation together with \Cref{lem:isometry_schur}.

\section{Miscellaneous facts}\label{sec:miscellaneous}
The proof of the following, which we omit, is an easy consequence of the cyclic property of the trace, applied to partial traces.
\begin{lemma}\label{lem:partial_trace_invariance}
  Let $X\in \mathrm{L}(\calH_\reg{A}\otimes \calH_{\reg{B}})$ and $D,E\in\mathrm{L}(\calH_\reg{B},\calH_{C})$. It holds that
  \begin{align*}
    \tr_{\reg{C}}\left((I_\reg{A}\otimes D)X(I_\reg{A}\otimes E^\dag)\right) = \tr_{\reg{B}}((I_\reg{A}\otimes E^\dag D)X).
  \end{align*}
\end{lemma}

The following is a straightforward consequence of Schur's lemma. Let $G$ be a compact group, and let $(\calH_{\reg{A}},r_\reg{A})$ and $(\calH_{\reg{B}},r_\reg{B})$ be finite-dimensional unitary representations of $G$. Let $\widehat G$ be a set of labels for a complete set of inequivalent irreducible representations of $G$, denoted $\{(\mathcal H^{(\lambda)},r^{(\lambda)}):\lambda\in\widehat G\}$.
As group representations, $\calH_{\reg{A}}$ and $\calH_{\reg{B}}$ have orthogonal decompositions
\begin{align*}
  \calH_{\reg{A}}
  &\cong
  \bigoplus_{\lambda\in\widehat G}
  \mathbb C^{k^{\reg{A}}_\lambda}\otimes \mathcal H^{(\lambda)}, \qquad
  \calH_{\reg{B}}
  \cong
  \bigoplus_{\lambda\in\widehat G}
  \mathbb C^{k^\reg{B}_\lambda}\otimes \mathcal H^{(\lambda)}.
\end{align*}
Denote these unitary relations by $T_\reg{A}$ and $T_\reg{B}$, so that
\begin{align*}
  T_\reg{A} r_\reg{A}(g)T_\reg{A}^\dag
  &=
  \bigoplus_{\lambda\in\widehat G}
  I_{k^\reg{A}_\lambda}\otimes r^{(\lambda)}(g), \qquad
  T_\reg{B} r_\reg{B}(g)T_\reg{B}^\dag
  =
  \bigoplus_{\lambda\in\widehat G}
  I_{k^\reg{B}_\lambda}\otimes r^{(\lambda)}(g).
\end{align*}
for every $g\in G$.
\begin{lemma}[Block diagonalization of intertwiners]\label{lem:intertwiner_implies_blocks}
Let $M\in\mathrm{L}(\calH_\reg{A},\calH_\reg{B})$ satisfy
$
  M r_\reg{A}(g) = r_\reg{B}(g)M
$
for every $g\in G$. Then
$
  T_\reg{B} M T_\reg{A}^\dag
  =
  \bigoplus_{\lambda\in\widehat G}
  M^{(\lambda)}\otimes I_{\mathcal H^{(\lambda)}},
$
where
$
  M^{(\lambda)}
  \in
  \mathcal L(\mathbb C^{k^\reg{A}_\lambda},\mathbb C^{k^\reg{B}_\lambda})
$.
\end{lemma}

\begin{proof}
Let
$
  \widetilde M
  \coloneq
  T_\reg{B} M T_\reg{A}^\dag.
$
which maps $\bigoplus_{\lambda\in\widehat{G}}\C^{k^\reg{A}_\lambda}\otimes \calH^{(\lambda)}$ to $\bigoplus_{\lambda\in\widehat{G}}\C^{k^\reg{B}_\lambda}\otimes \calH^{(\lambda)}$.
By the intertwining assumption, for every $g\in G$ we have
\begin{align*}
  \widetilde M
  \left(
    \bigoplus_{\lambda\in\widehat G}
    I_{k^\reg{A}_\lambda}\otimes r^{(\lambda)}(g)
  \right)
  =
  \left(
    \bigoplus_{\lambda\in\widehat G}
    I_{k^\reg{B}_\lambda}\otimes r^{(\lambda)}(g)
  \right)
  \widetilde M .
\end{align*}
Let
$
  \widetilde M^{\lambda\lambda'}
  \colon
  \mathbb C^{k^\reg{A}_{\lambda'}}\otimes\mathcal H^{(\lambda')}
  \to
  \mathbb C^{k^\reg{B}_\lambda}\otimes\mathcal H^{(\lambda)}
$
denote the $(\lambda,\lambda')$ block of $\widetilde M$. Then we have the matrix equality
\begin{align*}
  \widetilde M^{\lambda\lambda'}
  \left(
    I_{k^\reg{A}_{\lambda'}}\otimes r^{(\lambda')}(g)
  \right)
  =
  \left(
    I_{k^\reg{B}_\lambda}\otimes r^{(\lambda)}(g)
  \right)
  \widetilde M^{\lambda\lambda'}
\end{align*}
for every $g\in G$. This in turn implies that, for every
$i\in[k^\reg{B}_\lambda]$, $j\in[k^\reg{A}_{\lambda^\prime}]$, and $g\in G$ we have
\begin{align*}
  \widetilde M^{\lambda\lambda'}_{ij} r^{(\lambda')}(g)
  =
  r^{(\lambda)}(g)\widetilde M^{\lambda\lambda'}_{ij}
\end{align*}
where we have defined $\widetilde{M}^{\lambda\lambda^\prime}_{ij}:=(\bra{i}\otimes I_{\calH^{(\lambda)}})\widetilde{M}^{\lambda\lambda^\prime}(\ket{j}\otimes I_{\calH^{(\lambda^\prime)}})$.
By Schur's lemma (see, e.g., Lemma~4.1.4 in \cite{goodman2009symmetry}), for all such $i$, $j$, and $g$ there is a complex number $\alpha^{(\lambda)}_{ij}\in\mathbb C$ such that
$
  \widetilde M^{\lambda\lambda^\prime}_{ij}
  =
  \delta_{\lambda\lambda^\prime}\alpha^{(\lambda)}_{ij} I_{\mathcal H^{(\lambda)}}.
$
It follows that $\widetilde{M}$ is of the desired form with $(M^{(\lambda)})_{ij}=\alpha^{(\lambda)}_{ij}$.
\end{proof}

The special case when $\calH=\calH_\reg{A}=\calH_{\reg{B}}$ is perhaps more standard, and we record it explicitly below. Let $(\calH,r)$ be a finite-dimensional unitary representation of $G$. As a group representation, $\calH$ admits an orthogonal decomposition of the form
$\calH\cong \bigoplus_{\lambda\in \widehat{G}}\mathbb{C}^{k_\lambda}\otimes \calH^{(\lambda)}$. Denote this unitary isomorphism by $T$, so that $Tr(g)T^{\dag}=\bigoplus_{\lambda\in \widehat{G}} I_{k_{\lambda}}\otimes r^{(\lambda)}(g)$ for every $g\in G$. Note that if $G=\mathfrak{S}_n\times \mathrm{U}(d)$ then $T=U^d_\schur$.
\begin{lemma}[Block diagonalization of commuting operators]\label{lem:commuting_implies_blocks}
   Let $M\in \mathrm{L}(\calH)$ be a linear operator such that $r(g)$ commutes with $M$ for every
   $g\in G$. Then
   $T M T^{\dag}=\bigoplus_{\lambda\in \widehat{G}} M^{(\lambda)}\otimes I_{\calH^{(\lambda)}}$
   where for each $\lambda\in \widehat{G}$ we have $M^{(\lambda)}\in  \mathrm{L}(\mathbb{C}^{k_{\lambda}})$.
\end{lemma}

\end{document}